\documentclass[english]{scrartcl}
\usepackage[T1]{fontenc}
\usepackage{babel}

\usepackage{color}
\usepackage{amsmath,amssymb,amsthm,esint}
\usepackage{graphicx, float}

\usepackage{natbib}
\usepackage{todonotes}

\usepackage{lineno}\usepackage{lmodern}
\usepackage{geometry}
\usepackage{graphicx}

\numberwithin{equation}{section}
\numberwithin{figure}{section}

\theoremstyle{plain}
\newtheorem{thm}{\protect\theoremname}
\newtheorem{lemma}[thm]{Lemma}
\newtheorem{corollary}[thm]{Corollary}
\newtheorem{remark}[thm]{Remark}
\providecommand{\theoremname}{Theorem}

\DeclareOldFontCommand{\rm}{\normalfont\rmfamily}{\mathrm}
\DeclareOldFontCommand{\sf}{\normalfont\sffamily}{\mathsf}
\DeclareOldFontCommand{\tt}{\normalfont\ttfamily}{\mathtt}
\DeclareOldFontCommand{\bf}{\normalfont\bfseries}{\mathbf}
\DeclareOldFontCommand{\it}{\normalfont\itshape}{\mathit}
\DeclareOldFontCommand{\sl}{\normalfont\slshape}{\@nomath\sl}
\DeclareOldFontCommand{\sc}{\normalfont\scshape}{\@nomath\sc}
\DeclareRobustCommand*\cal{\@fontswitch\relax\mathcal}
\DeclareRobustCommand*\mit{\@fontswitch\relax\mathnormal}

\newcommand{\bsa}{{\mathbf{a}}}
\newcommand{\bsb}{{\mathbf{b}}}
\newcommand{\bsx}{{\mathbf{x}}}  
\newcommand{\bsy}{{\mathbf{y}} }  
\newcommand{\bsz}{{\mathbf{z}} } 
\newcommand{\bsv}{{\mathbf{v}} } 
\newcommand{\bsw}{{\mathbf{w}} } 
\newcommand{\bsu}{{\mathbf{u}} } 
\newcommand{\bsmu}{{\boldsymbol{\mu}} } 
\newcommand{\bsalpha}{{\boldsymbol{\alpha}} }
\newcommand{\bsbeta}{{\boldsymbol{\beta}} }
\newcommand{\bstheta}{{\boldsymbol{\theta}}}

\newcommand{\bsC}{{\boldsymbol{C}}}  

\newcommand{\boldzero}{\mathbf{0}}
\newcommand{\rd}{\mathrm{d}}	
\newcommand{\N}{\mathbb{N}}	
\newcommand{\R}{\mathbb{R}}

\newcommand{\cO}{\mathcal{O}}

\newcommand{\convprob}{\overset{p}{\rightarrow}}
\newcommand{\samplespace}{\mathcal{Z}}
\newcommand{\plimn}{\plim_{n\rightarrow\infty}}

\DeclareMathOperator*{\plim}{plim}
\DeclareMathOperator*{\argmax}{arg\,max}

\begin{document}

\title{Maximum Approximated Likelihood Estimation}

\author{M. Griebel \and F. Heiss \and J. Oettershagen \and C. Weiser}
\maketitle

\begin{abstract}
	Empirical economic research frequently applies maximum likelihood estimation in cases where the likelihood function is analytically intractable. Most of the theoretical literature focuses on maximum simulated likelihood (MSL) estimators, while empirical and simulation analyzes often find that alternative approximation methods such as quasi--Monte Carlo simulation, Gaussian quadrature, and integration on sparse grids behave considerably better numerically. This paper generalizes the theoretical results widely known for MSL estimators to a general set of \emph{maximum approximated likelihood (MAL)} estimators. We provide general conditions for both the model and the approximation approach to ensure consistency and asymptotic normality. We also show specific examples and finite--sample simulation results.
\end{abstract}


\section{Introduction}
Consider classical maximum-likelihood estimation, i.e. the estimated parameter vector \(\hat{\bstheta}_{ML} \in \Theta \subset \R^p\) is obtained by maximizing
\begin{equation} \label{eqn_intro_likelihood}
 L_n(\bstheta) = \frac{1}{n} \sum_{i=1}^n \log f(\bsz_i, \bstheta) .
\end{equation}
Here, \(f(\bsz_i, \bstheta)\) denotes the individual likelihood contribution of the sample \(\bsz_i \in \samplespace \subset \R^q\).
As \(n\) tends to infinity, the ML-estimator converges to the true parameter \(\bstheta_0\), which defines the distribution the sample \(\bsz_1,\ldots, \bsz_n\) was drawn from, i.e. \(\plim_{n \to \infty} \hat{\bstheta}_{ML} = \bstheta_0\).

Often, the function \(f\) stems from an integral representation
\begin{equation} \label{eqn_intro_integeral}
 f(\bsz, \bstheta) = \int_\Omega \varphi(\bsv, \bsz, \bstheta) \, \omega(\bsv) \rd \bsv ,
\end{equation}
where \(\varphi: \Omega \times \samplespace \times \Theta \to \R\) is some (usually non-negative) function that is integrated with respect to the variable \(\bsv\) and a weight function \(\omega: \Omega \to \R_+\) that is defined on a domain \(\Omega \subset \R^d\).
In many relevant cases, these integrals cannot be computed in closed form, e.g. for models of discrete choice \citep{Butler1982, McFaddenTrain2000, Train2009}, or general limited dependent variable models \citep{Hajivassiliou1994}.
Then, \(f\) can be approximated with an \(r\)-point quadrature rule, i.e.
\begin{equation} \label{eqn_intro_quadrature}
 \tilde{f}_r(\bsz, \bstheta) := \sum_{j=1}^r w_{j,r} \varphi(\bsv_{j,r}, \bsz, \bstheta) \quad \approx f(\bsz, \bstheta) .
\end{equation}
Here, the quadrature points and weights \((\bsv_{j,r}, w_{j,r})_{j=1}^r\) are independed of $\bsz$ and need to be chosen properly to guarantee a certain accuracy, provided that specific assumptions on \(\varphi\) are valid.
Moreover, increasing the parameter \(r\) allows to increase the accuracy on the one hand, but leads to additional cost on the other hand. Therefore, it is necessary to determine the required accuracy for the approximation of \(f\) by \(\tilde{f}_r\) such that the resulting estimator maintains consistency and asymptotical normality.

To this end, we introduce a link function \(R: \N \to \N\) that couples the number of integration points \(r\) to the sample size \(n\), i.e. \(r = R(n)\). Now, \(R(n)\) can be chosen such that \(\tilde{f}_{R(n)}(\bsz_i, \bstheta)\) approximates \(f(\bsz_i, \bstheta)\) well enough to ignore this additional approximation error but not better, to keep the overall cost at a tractable level.
Altogether, this leads to the \emph{maximum approximated likelihood estimator} (MALE) given by
\begin{equation}
 \hat{\bstheta}^{(n)}_{MAL} := \argmax_{\bstheta \in \Theta} \tilde{L}_n(\bstheta), \quad \text{ where } \tilde{L}_n(\bstheta) := \frac{1}{n} \sum_{i=1}^n \log \tilde{f}_{R(n)}(\bsz_i, \bstheta) .
\end{equation}
Here, the cost for one evaluation of \(\tilde{L}_n\) in terms of evaluations of \(\varphi\) is \(n \cdot R(n)\).

A classical and flexible tool to deal with the numerical integration problems in \eqref{eqn_intro_quadrature} is Monte Carlo simulation (MC). Here, all weights are chosen uniformly as \(w_{j,r} = 1/r, j=1,\ldots,r\) and the points \(\bsv_{j,r}\) are drawn identically and independently from the probability distribution induced by the weight function \(\omega\). Since MC only requires weak assumptions on \(\varphi\), simulation-based estimation has become part of the standard econometrics toolkit. It is implemented in software packages and taught at graduate schools.
Comprehensive surveys can be found in \cite{Hajivassiliou1994}, \cite{GourierouxMonfort1996} and \cite{Train2009}.

Maximum simulated likelihood (MSL) is consistent under the usual assumptions if the number of simulation draws \(r\) increases with the sample size \(n\). In order to achieve asymptotic efficiency for identical sampling, \(r\) has to grow at least linearly in \(n\), i.e. \(r = R(n) \geq n\).\footnote{"Identical sampling" denotes the use of one draw-set for all likelihood contributions, i.e. the same quadrature rule is used for all samples $\bsz_i$. In contrast, "independent sampling" denotes the use of one draw-set per likelihood contribution, i.e. for each likelihood contribution an individual quadrature rule is applied. In case of independent sampling, \(r\) only has to increase faster than \(\sqrt{n}\), see for example \cite{Hajivassiliou1994}. A recent discussion about the difference between both methodologies can also be found in \cite{kristensen2017}.}

While asymptotically the simulation error disappears under the appropriate conditions, the computational burden imposed by a sufficient number of simulation draws to achieve well-behaved estimators can be cumbersome or prohibitive in practice. Many studies have found that approximation algorithms other than Monte-Carlo simulation can achieve much higher accuracy and better behaved estimators with a reduced level of computational costs. For one-dimensional problems that are sufficiently smooth, Gaussian quadrature \citep{Butler1982} is an obvious choice. For multivariate problems, quasi-Monte-Carlo methods like Halton draws \citep{Bhat2001} or other deterministic approaches like integration on sparse grids \citep{Heiss2008} have been applied successfully.

These deterministic approximation algorithms work well enough in practice to warrant their routine use in software packages like the mixed logit implementation of Stata. But while there is plenty of literature on properties of simulation-based approaches, little is known for estimators based on the other methods.
Some papers that use deterministic numerical integration methods other than pure simulation ignore the approximation error in the discussion of the estimator. This applies mainly to examples with one-dimensional integration problems which are tackled with Gaussian quadrature such as \cite{Butler1982}.
Other papers discuss the well-studied simulation-based estimators before arguing that their deterministic approaches tend to work better in practice, e.g. in \cite{Bhat2001} or in \cite{sandor2004}.
This might be due to the fact that the theoretical properties and prerequisites of deterministic approximation schemes are not yet very well understood.
\cite{Ackerberg2009} took a first step by giving a set of conditions for the approximated log likelihood contributions that imply consistency and asymptotic normality. However their results remain on the rather abstract level of log-likelihood approximation and provide no guidance  on how to check these conditions for specific applications and approximation algorithms. In particular, it remains unclear how the integration error for the approximation of \eqref{eqn_intro_integeral} by \eqref{eqn_intro_quadrature} influences the statistical properties of the approximated likelihood estimator. 

This paper aims at closing this gap and provides a general and unified discussion of the asymptotic properties of a large class of estimators that is based on a broad range of integration methods. The well-known results of simulation-based approaches emerge as special cases. We provide specific conditions under which these estimators are consistent and under which additional conditions the approximation error is irrelevant for their asymptotic distribution.
For example, we can derive that a logarithmic growth of the number of samples \(r\) is sufficient, if the convergence in \(r\) is fast enough. Therefore, in the setting of Gaussian quadrature, it suffices to have \(r = R(n) = \log(n)\). This is a huge reduction compared to \(R(n) \approx n\) or \(R(n) \approx \sqrt{n}\).
As an application and demonstration of our framework to a specific model, we deal with mixed logit models and the Butler-Moffitt model, which both are estimated using an approximation that is based on Gauss-Hermite quadrature or Gauss-Hermite sparse grids.

The remainder of this article is organized as follows: In Section \ref{sec_asymp_properties} we discuss conditions that imply the consistency and asymptotic normality of general extremum estimators.
In Section \ref{sec_tildef} we specialize on maximum approximated likelihood estimators and the assumptions that have to be made such that the conditions from the previous section are fulfilled.
Then, in Section \ref{s_rules}, we deal with specific integration algorithms, like Gaussian quadrature, quasi--Monte Carlo and sparse grids. In Section \ref{s_examples}, we put our theory into practice, by analyzing specific econometric models in the light of our results. This analysis is supplemented by numerical results in Section \ref{sec_pre_asymp}.

\section{Asymptotic theory with approximated objective functions} \label{sec_asymp_properties}

As the most general framework, we consider extremum or M-estimators as discussed in \citet{Newey1994}, NM hereafter. An M-estimator \(\hat{\bstheta}_{M}\) maximizes (or minimizes) some objective function \(Q_n(\bstheta)\), i.e.
\begin{equation}
\hat{\bstheta}_{M} = \arg\max_{\bstheta\in\Theta} Q_{n}(\bstheta) ,
\end{equation}
where $n$ refers to the number of samples $\bsz_1,\ldots, \bsz_n$ that contribute to \(Q_n(\bstheta)\).
Examples include least squares, maximum likelihood, GMM, and minimum distance.
Throughout this paper, we will maintain the assumption that \(\hat{\bstheta}_{M}\) would have all desired properties if we were able to compute the respective objective function \(Q_{n}(\bstheta)\) and therefore the estimator itself. NM comprehensively discuss the conditions to ensure these properties.

We are interested in the problems that arise if \(Q_n(\bstheta)\) cannot be evaluated analytically and therefore needs to be approximated. Let \(\tilde Q_n(\bstheta)\) denote the approximated objective function. Then, our approximated M-estimator is simply
\begin{equation}
\hat{\bstheta}_{AM} = \arg\max_{\bstheta\in\Theta} \tilde Q_{n}(\bstheta) .
\end{equation}
We now give general conditions to ensure consistency and the asymptotic distribution of
\(\hat{\bstheta}_{AM}\).
We start by recalling Theorem 2.1 of NM:

\begin{lemma} \label{extremum_estimator-consist}
  Assume that there is a function \(Q_0(\bstheta)\) such that
  (i) \(Q_0(\bstheta)\) is uniquely maximized at \(\bstheta_0\);
  (ii) \(\bstheta \in \Theta\) and \(\Theta\) is compact;
  (iii) \(Q_0(\bstheta)\) is continuous;
  (iv) the function \(Q_n(\bstheta)\) converges uniformly in probability to \(Q_0(\bstheta)\).
  Then, \(\hat{\bstheta}_M = \arg\max_{\bstheta \in \Theta}Q_n(\bstheta) \) is a consistent estimator of \(\bstheta_0\), i.e.
\(  \plimn \hat{\bstheta} = \bstheta_{0}.
\)
\end{lemma}

We abstract from any misspecifications and other problems that could violate the conditions of Lemma \ref{extremum_estimator-consist} in order to focus on the inaccuracies introduced by the approximation of the objective function. Indeed, We can apply the same arguments as in Lemma \ref{extremum_estimator-consist} to the approximated M-estimator \(\hat\bstheta_{AM}\) if we can ensure that 
assumption (iv) also holds for the approximated objective function. The following Theorem \ref{thm_consist_L} states that this is the case as long as \(\tilde{Q}_{n}(\bstheta)\) converges uniformly in probability to the unavailable exact objective function \(Q_{n}(\bstheta)\).

\begin{thm}\label{thm_consist_L}
Assume that
 (i) the assumptions of Lemma \ref{extremum_estimator-consist} hold
 (ii) \(\tilde Q_n(\bstheta)\) converges uniformly in probability to \(Q_n(\bstheta)\), i.e. \(\plimn \sup_{\bstheta\in\Theta} \left|\tilde{Q}_{n}(\bstheta)-Q_{n}(\bstheta)\right| = 0.\)
Then, \(\hat{\bstheta}_{AM}\) is a consistent estimator of \(\bstheta_0\), i.e.
\begin{equation}
\plimn \hat{\bstheta}_{M} = \bstheta_{0}.
\end{equation}
\end{thm}

\begin{proof}
We assume in \ref{thm_consist_L}(i) that the M-estimator with the exact objective function is consistent. We can use the same arguments as NM to show Lemma \ref{extremum_estimator-consist} if we can establish that \(\tilde Q_n(\bstheta)\) converges uniformly in probability to \(Q_0(\bstheta)\) (Assumption \ref{extremum_estimator-consist}(iv)). 
To see why this holds true, note that, by the triangle inequality for norms there holds

\begin{align}\label{ErrorComposition}
\sup_{\bstheta\in\Theta}\left|\tilde{Q}_{n}(\bstheta)-Q_{0}(\bstheta)\right|
& =
\sup_{\bstheta\in\Theta}\left|\big(\tilde{Q}_{n}(\bstheta)-Q_{n}(\bstheta)\big)+
\big(Q_{n}(\bstheta)-Q_{0}(\bstheta)\big)\right| \\
& \leq
\sup_{\bstheta\in\Theta}\left|\tilde{Q}_{n}(\bstheta)-Q_{n}(\bstheta)\right|+
\sup_{\bstheta\in\Theta}\left|Q_{n}(\bstheta)-Q_{0}(\bstheta)\right|.
\nonumber
\end{align}
Both terms converge to zero in probability: the first by assumption \ref{thm_consist_L}(ii), the second by assumption \ref{thm_consist_L}(i).
\end{proof}

Note that assumption \ref{thm_consist_L}(ii) requires the approximation accuracy to increase with \(n\). In the well-known example of approximation by Monte Carlo simulation, we can increase the number of simulation draws as \(n\rightarrow\infty\). We will come back to this more explicitly when we discuss specific approximations approaches.

accuracy parameter \(r\) to the number of observations \(n\) we introduce a function \(R: \N \to \N\). We will assume that \(R(n)\) is monotonically increasing, i.e. \(R(n+1) \geq R(n)\).

In order to derive the asymptotic distribution of \(\hat\bstheta_{AM}\), we again recall the relevant Theorem 3.1 of NM for extremum estimators \(\hat{\bstheta} = \argmax_\bstheta {Q}_n\).

\begin{lemma} \label{extremum_estimator-normal} 
  Suppose that the assumptions of Lemma \ref{extremum_estimator-consist} hold and that
  (i) \(\bstheta_{0} \in \text{interior}(\Theta)\);
  (ii) \(Q_n(\bstheta)\) is twice continuously differentiable in a neighborhood \(\mathcal{N}\) of \(\bstheta_{0}\);
  (iii) \(\sqrt{n}\nabla_{\bstheta}Q_n(\bstheta_{0}) \overset{d}{\rightarrow} N(0,\Sigma)\);
  (iv) there is \(H(\bstheta)\) that is continuous at \(\bstheta_0\) and \(\sup_{\bstheta\in\mathcal{N}}\left\Vert \nabla_{\bstheta\bstheta} Q_n(\bstheta)-H(\bstheta)\right\Vert \overset{p}{\rightarrow}0\);
  (v) \(H=H(\bstheta_0)\) is nonsingular. Then the M-estimator is asymptotically normal distributed, i.e
\(
  \sqrt{n}(\hat{\bstheta}-\bstheta_0) \overset{d}{\rightarrow} N(0,H^{-1}\Sigma H^{-1})
\).
\end{lemma}

To analyze our approximate M-estimator, we again assume that the conditions of this lemma hold and then give additional assumptions such that the same results apply to the approximate M-estimator \(\hat{\bstheta}_{AM}\).

\begin{thm}\label{thm_asynorm_L}
Assume that (i) the assumptions of Lemma \ref{extremum_estimator-normal} hold;
(ii) With probability one, \(\tilde{Q}_n(\bstheta)\) is twice continuously differentiable in a neighborhood \(\mathcal{N}\) of \(\bstheta_{0}\);
(iii) \(\plimn \sup_{\bstheta\in\Theta} \left|\tilde{Q}_{n}(\bstheta)-Q_{n}(\bstheta)\right| = 0\);
(iv) \(\plimn \sqrt{n}\,\sup_{\bstheta\in\Theta} \left\Vert\nabla_{\bstheta}\tilde{Q}_{n}(\bstheta) - \nabla_{\bstheta}Q_{n}(\bstheta)\right\Vert = 0\);
(v) \(\plimn \sup_{\bstheta\in\Theta} \left\Vert\nabla_{\bstheta\bstheta}\tilde{Q}_{n}(\bstheta) - \nabla_{\bstheta\bstheta}Q_{n}(\bstheta)\right\Vert = 0\).
Then
\[
\sqrt{n}\left(\hat{\bstheta}_{AM}-\bstheta_{0}\right)\overset{d}{\rightarrow}N(\boldzero,H^{-1}\Sigma H^{-1}), \quad (\text{as } n \to \infty)
\]
so \(\hat{\bstheta}_{AM}\) has the same limiting distribution as the infeasible \(\hat{\bstheta}_{M}\).
\end{thm}

\begin{proof}
We can apply the same arguments used by NM to show Lemma \ref{extremum_estimator-normal}, applied to our approximated objective function \(\tilde Q_n(\bstheta)\). By our assumptions \ref{thm_asynorm_L}(i) and \ref{thm_asynorm_L}(ii), assumptions \ref{extremum_estimator-normal}(i), \ref{extremum_estimator-normal}(ii), and \ref{extremum_estimator-normal}(v) are directly implied. To check \ref{extremum_estimator-normal}(iii), we write
\begin{align}
\sqrt{n}\nabla_{\bstheta}\tilde Q_n(\bstheta_{0}) &=
\sqrt{n}\nabla_{\bstheta}Q_n(\bstheta_{0}) +
\sqrt{n}\left(\nabla_{\bstheta}\tilde Q_n(\bstheta_{0}) -
\nabla_{\bstheta}Q_n(\bstheta_{0})\right) .
\end{align}
The first term converges in distribution to \(N(0,\Sigma)\) by assumption \ref{thm_asynorm_L}(i) and \ref{extremum_estimator-normal}(iii). The second term converges in probability to zero by \ref{thm_asynorm_L}(iv). Finally, we confirm \ref{extremum_estimator-normal}(iv) by noting that
\begin{align}
\sup_{\bstheta\in\mathcal{N}}\big\Vert \nabla_{\bstheta\bstheta} \tilde Q_n(\bstheta)-H(\bstheta)\big\Vert
&\leq
\sup_{\bstheta\in\mathcal{N}}\big\Vert \nabla_{\bstheta\bstheta} \tilde Q_n(\bstheta)-\nabla_{\bstheta\bstheta}Q_n(\bstheta)\big\Vert +
\sup_{\bstheta\in\mathcal{N}}\big\Vert \nabla_{\bstheta\bstheta} Q_n(\bstheta)-H(\bstheta)\big\Vert
\end{align}
is implied by the triangle inequality. Both terms converge to zero in probability: the first by assumption \ref{thm_asynorm_L}(v), the second by assumption \ref{thm_asynorm_L}(i).
\end{proof}

Arguably the most important condition for deriving the asymptotic distribution is assumption \ref{thm_asynorm_L}(iv). It not only requires the approximated gradient to converge to the exact value, but the rate of convergence also needs to be faster than \(1/\sqrt n\).

\section{Maximum approximated likelihood} \label{sec_tildef}

In the last section we discussed general extremum estimators with approximated objective functions. These results are similar to those of \cite{Ackerberg2009} who study general approximation algorithms for maximum likelihood estimation. For the remainder of this paper, we focus on the case of maximum likelihood. This covers a large share of the applications of approximate M-estimation and allows us to be more specific. We will derive general conditions for the likelihood contributions and the approximation algorithm to ensure favorable properties of the resulting MAL estimator.

\subsection{Asymptotic analysis with respect to likelihood contributions}

We consider an i.i.d. random sample \([\bsz_i\in \samplespace;\ i=1,\dots,n]\) from a population distribution characterized by the family of probability mass or density functions \(f(\bsz; \bstheta_0)\) and the sample space \(\samplespace \subset \R^d\). Here, \(\bsz\) includes all variables. In most econometric models, \(\bsz\) includes some ``endogenous'' variables \(\bsy\) and some ``exogenous'' variables \(\bsx\). In these cases, \(f(\bsz; \bstheta_0)\) is actually conditional on \(\bsx\). For notational convenience and consistency with the literature, we will not explicitly make this distinction. The exact log likelihood function is
\begin{equation}\label{eq:llf}
Q_n(\bstheta)=L_n (\bstheta) = \frac 1 n \sum_{i=1}^n \log f(\bsz_i; \bstheta),
\end{equation}
where the individual likelihood contributions of sample \(\bsz_i\) are given by \(f(\bsz_i; \bstheta)\).
The maximum likelihood estimator \(\hat{\bstheta}_{ML}\) maximizes \(L_n (\bstheta)\), i.e. \(\hat{\bstheta}_{ML} = \argmax_\bstheta L_n (\bstheta)\).
Moreover, in the context of the preceding section, we have \(\plim_{n \to \infty} Q_n(\bstheta) = Q_0(\bstheta)\), where \(Q_0\) was introduced in Lemma \ref{extremum_estimator-consist} and is maximized by the true parameter \(\bstheta_0\).

\citet[Theorems 2.5, 3.3]{Newey1994} provide conditions for the ML estimator to have favorable properties like consistency and asymptotic normality. We recall them in the following, treating the consistency of the ML estimator first.
\begin{lemma}\label{thm_ML_consist}
Assume that (i) For all \(\bstheta\neq\bstheta_{0}\), we have \(f(\bsy;\bstheta)\neq f(\bsy;\bstheta_{0})\);
(ii) \(\bstheta_{0}\in\Theta\) and \(\Theta\) is compact;
(iii) \(\log f(\bsy;\bstheta)\) is continuous at each \(\bstheta\in\Theta\) with probability one;
(iv) \(\mathbb{E}_\bsy[\sup_{\bstheta\in\Theta}|\log f(\bsy;\bstheta)|]<\infty\).
Then, \(\plim_{n\rightarrow \infty} \hat{\bstheta}_{ML}=\bstheta_{0}\) (consistency).
\end{lemma}

Next, we recall the following result on the asymptotic normality of the ML estimator.

\begin{lemma}\label{thm_ML_distr}
Assume that (i) the assumptions of Lemma \ref{thm_ML_consist} hold and that
(ii) \(\bstheta_{0}\in\text{interior}(\Theta)\);
(iii) \(f(\bsy;\bstheta)\) is twice continuously differentiable and \(f(\bsy;\bstheta)>0\) in a neighborhood \(\mathcal{N}\) of \(\bstheta_{0}\);
(iv) \(\int\sup_{\boldsymbol{\theta\in\mathcal{N}}}\left\Vert
\nabla_{\bstheta}f(\bsy;\bstheta)\right\Vert
\,d\bsy<\infty\)
and
\(\int\sup_{\boldsymbol{\theta\in\mathcal{N}}}\left\Vert
\nabla_{\bstheta\bstheta}f(\bsy;\bstheta
)\right\Vert \,d\bsy<\infty\);
(v)
\(\mathcal{I}=\mathbb{E}_\bsy \left[\nabla_{\bstheta}f(\bsy;\bstheta
_{0})\left(\nabla_{\bstheta}f(\bsy;\bstheta_{0}
)\right)^{T}\right]\)
exists and is nonsingular,
(vi) \(\mathbb{E}_\bsy \left[\sup_{\boldsymbol{\theta\in\mathcal{N}}}\left\Vert
\nabla_{\bstheta}\log
f(\bsy;\bstheta_{0})\right\Vert \right]<\infty\).
Then
\(\sqrt{n}\left(\hat{\bstheta}_{ML}-\bstheta_{0
}\right)\overset{d}{\rightarrow}N\left(0,\mathcal{I}^{-1}\right)\) (asymptotic normality).
\end{lemma}

If the likelihood contributions \(f(\bsz_i; \bstheta)\) cannot be computed analytically, we need to approximate them using some algorithm. The approximated likelihood contributions are denoted by \(\tilde f_r(\bsz_i; \bstheta), i=1,\dots ,n\) and depend on an accuracy parameter \(r\), but not on the sample $\bsz_i$, i.e. the same approximation approach to $f$ is used for all likelihood contributions.
For maximum simulated likelihood, \(r\) might be the number of simulation draws. In general, it determines the amount of approximation error as well as the computational costs.

Theorems \ref{thm_consist_L} and \ref{thm_asynorm_L} require the approximated objective function \(\tilde Q_n(\bstheta)\) to converge to the exact function \(Q_n(\bstheta)\) as \(n\rightarrow\infty\). In general, this requires \(r\) to increase with \(n\). Therefore, it is not sufficient to keep \(r\) fixed for all \(n\), but to consider a sequence of approximation functions \((\tilde{f}_{r}(\bsz_{i};\bstheta))_{r=1}^\infty\) which, in a sense that will be clarified later, converge to the exact \(f(\bsz_{i};\bstheta)\) as \(r\rightarrow\infty\).
Then, the approximated log-likelihood
\begin{equation}
 \tilde{L}_{n,r}(\bstheta) := \frac{1}{n} \sum_{i=1}^n \log \tilde{f}_r(\bsz_i; \bstheta)
\end{equation}
fulfills \(\lim_{n,r \to \infty} \tilde{L}_{n,r}(\bstheta) = Q_0(\bstheta)\). In practice for any finite \(n \in \N\), we need to choose a finite accuracy level \(r\). Choosing \(r\) too large will result in unneccesary cost, while using a too small \(n\) will result in an additional error contribution. This trade-off will be the subject of the remainder of this section.

In order to link the accuracy parameter \(r\) to the number of observations \(n\) we introduce a function \(R: \N \to \N\). We will assume that \(R(n)\) is monotonically increasing, i.e. \(R(n+1) \geq R(n)\) and \(\lim_{n \to \infty} R(n) = \infty\). The speed with which \(R(n)\) needs to increase with \(n\) will depend on the specific approximation algorithms as will be discussed below.

Then, the maximum approximated likelihood estimator \(\hat{\bstheta}_{MAL}\) maximizes
\begin{equation} 
\tilde Q_n(\bstheta)=\tilde{L}_n (\bstheta) = \frac 1 n \sum_{i=1}^n \log \tilde{f}_{R(n)}(\bsz_i; \bstheta). \label{eq_Ltilde}
\end{equation}

We now discuss general properties of \(\tilde{f}_{R(n)}(\bsz_i; \bstheta)\) that imply the consistency and the asymptotic distribution of
\(\hat\bstheta_{MAL}\). In order to focus on the effects of approximation, we assume that the infeasible ML estimator would have all desired properties if it were available, i.e. that the assumptions of Lemma \ref{thm_ML_consist} and \ref{thm_ML_distr} hold.

In addition to the consistency of the ML estimator, the only additional assumption we need for consistency of \(\hat{\bstheta}_{MAL}\) according to Theorem \ref{thm_consist_L} is the uniform convergence in probability of the approximation \(\tilde{L}_{n}(\bstheta)\) to the true, but intractable \(L_{n}(\bstheta)\).

\begin{thm}\label{thm_consist_f}
Assume that
(i) the conditions of Lemma \ref{thm_ML_consist} hold and that
(ii) there exists \(\bar{\delta} > 0\) such that for all \(\bsz \in \samplespace \) and all \(\bstheta \in \Theta\) it holds that \(f(\bsz, \bstheta) \geq \bar{\delta}\);
(iii) \( \lim_{r \to \infty} \sup_{\bsz\in\samplespace,\bstheta\in\Theta}\left|\tilde{f}_{r}(\bsz, \bstheta) - f(\bsz, \bstheta)\right| = 0\);
(iv) \(R(n)\) is monotonically increasing in \(n\) and \(\lim_{n \to \infty} R(n) = \infty\).
Then, \(\hat{\bstheta}_{MAL}\) is a consistent estimator of \(\bstheta_0\), i.e.
\begin{equation}
\hat{\bstheta}_{MAL}\convprob \bstheta_{0}.
\end{equation}
\end{thm}

\begin{proof}
We use our general results from Theorem \ref{thm_consist_L}.
Given the (infeasible) MLE is consistent by assumption (i), we only have to verify Assumption \ref{thm_consist_L}(ii).

Note at this point that there exists \(r_0 \in \N\) and a real number \(\delta \in (0, \bar{\delta}) \) such that for all \(r \geq r_0\), all \(\bsz \in \samplespace \) and all \(\bstheta \in \Theta\) it holds that both, \(\tilde{f}_r(\bsz, \bstheta) \geq \delta\) and \(f(\bsz, \bstheta) \geq \bar{\delta} \geq \delta\). To see why this is true, we use that there exists \(\bar{\delta} > 0\) such that for all \(\bstheta \in \Theta\) and \(\bsz \in \samplespace\) it holds \(f(\bstheta, \bsz) > \bar{\delta}\). Then, we choose some \(\delta\) with \(0 < \delta < \bar{\delta}\). 
Because of \ref{thm_consist_f}(iii), we have \(\lim_{r \to \infty} \tilde{f}_r(\bsz;\bstheta) = f(\bsz; \bstheta)\) for all \(\bsz \in \samplespace \) and \(\bstheta \in \Theta\). Hence, there exists \(r_0(\delta) \in \N\) such that for all \(r \geq r_0(\delta)\) it holds for all \(\bsz \in \samplespace \) and all \(\bstheta \in \Theta\) that
\begin{equation} \label{eqn_bardelta}
| \tilde{f}_r(\bsz;\bstheta) - f(\bsz; \bstheta)| < \bar{\delta} - \delta .
\end{equation}
This implies \(\tilde{f}_r(\bsz;\bstheta) \geq \delta\) for all \(r \geq r_0(\delta)\).

Now we are in the position to show
\begin{equation}
\sup_{\bstheta\in\Theta}
\left|\tilde{L}_{n}(\bstheta)-L_{n}(\bstheta)\right|
\convprob 0.
\end{equation}
To this end, we write
\begin{align*}
 \sup_{\bstheta \in \Theta} \left| \tilde{L}_n(\bstheta) - L_n(\bstheta) \right| & = \sup_{\bstheta \in \Theta} \left| \frac{1}{n} \sum_{i=1}^n \log( \tilde{f}_{R(n)}(\bsz_i, \bstheta) ) - \log( f(\bsz_i, \bstheta) ) \right| \\
 & \leq  \sup_{\bstheta \in \Theta, \bsz \in \samplespace} \left| \log( \tilde{f}_{R(n)}(\bsz, \bstheta) ) - \log( f(\bsz, \bstheta) ) \right| .
 \end{align*}

For sufficiently large \(n\) (i.e. \(R(n) \geq r_0\)), we can use Corollary \ref{cor_concat} (i) in Appendix \ref{app_thm}. By assumptions (i) and (ii) with \(h(\bstheta) = f(\bsz_i, \bstheta)\) and \(g(\bstheta) = \tilde{f}_{R(n)}(\bsz_i, \bstheta)\) it follows that
 \begin{align*}
 \sup_{\bstheta \in \Theta} \left| \tilde{L}_n(\bstheta) - L_n(\bstheta) \right|& \leq 
 \frac{1}{\delta} \sup_{\bstheta \in \Theta, \bsz \in \samplespace} \left| \tilde{f}_{R(n)}(\bsz, \bstheta)-f(\bsz, \bstheta) \right| .
\end{align*}
Moreover, by assumption (iii) and (iv), this expression converges to zero in probability.
\end{proof}

In order to study the asymptotic distribution of the MAL estimator,
we assume that the (intractable) ML estimator \(\hat\bstheta_{ML}\) is asymptotically normally distributed and efficient and that the MAL estimator \(\hat\bstheta_{MAL}\) is consistent. Then we make additional assumptions to derive the asymptotic distribution of \(\hat\bstheta_{MAL}\). To this end, we define the quantity
\begin{equation}\label{eq_error}
 \mathcal{E}(r) :=\sup_{\bsz\in\samplespace,\bstheta\in\Theta} \left| \tilde{f}_{r}(\bsz; \bstheta) - f(\bsz; \bstheta) \right| + \sup_{\bsz\in\samplespace,\bstheta\in\Theta} \left\Vert\nabla_{\bstheta}\tilde{f}_{r}(\bsz; \bstheta) - \nabla_{\bstheta}{f}(\bsz; \bstheta)\right\Vert ,
\end{equation}
which measures the worst-case approximation error of both, the function \(f\) and its gradient \(\nabla f\) by \(\tilde{f}_r\) and \(\nabla \tilde{f}_r\), respectively.

\begin{thm}\label{thm_asynorm_f}
Assume that the assumptions of Lemma \ref{thm_ML_distr} and Theorem \ref{thm_consist_f} hold.
Also assume that
(i) \(\tilde{f}_r(\bsz; \bstheta)\) is twice continuously differentiable in \(\bstheta\in\Theta\) and bounded for all \(\bsz\in\samplespace\) and \(r\in\mathbb N\);
(ii) \(\plim_{n \to \infty } \sqrt{n}\,\mathcal{E}(R(n)) = 0\);
(iii) \(\sup_{\bsz\in\samplespace,\bstheta\in\Theta} \left\Vert\nabla_{\bstheta\bstheta}\tilde{f}_{R(n)}(\bsz; \bstheta)-\nabla_{\bstheta\bstheta}{f}(\bsz; \bstheta)\right\Vert \convprob 0\).
Then
\[
\sqrt{n}\left(\hat{\bstheta}_{MAL} - \bstheta_{0}\right)\overset{d}{\rightarrow}N(\boldzero,\mathcal{I}^{-1}).
\]
\end{thm}

\begin{proof}
We proceed by verifying the assumptions of Theorem \ref{thm_asynorm_L}. Assumption \ref{thm_asynorm_L}(i) directly follows from Assumption \ref{thm_asynorm_f}(i).

The critical conditions we have to check are \ref{thm_asynorm_L}(iv) and  \ref{thm_asynorm_L}(v). The arguments are similar to that in the proof of Theorem \ref{thm_consist_f}: Again, there exists a \(\delta \in (0, \bar{\delta})\) such that for sufficiently large \(r\) it holds \(\tilde{f}_r > \delta\).
In order to show \ref{thm_asynorm_L}(iv), we define
\[
 C_1(f) := \frac{1 + \sup_{\bsz, \bstheta} \|\nabla_\bstheta f(\bsz, \bstheta)\|}{\delta^2}
\]
and 
\[
    C_2(f) := 2\, \frac{1 + \sup_{\bsz, \bstheta} \|\nabla_\bstheta f(\bsz, \bstheta)\|^2 + \sup_{\bsz, \bstheta} \|\nabla_{\bstheta,\bstheta} f(\bsz, \bstheta)\| }{\delta^3} .
\]

 Then, we use Corollary \ref{cor_concat} (ii) and the definition of \(\mathcal{E}_r\) to derive
\begin{align*}
 \sqrt{n}\,\sup_{\bstheta\in\Theta}\left\Vert\nabla_{\bstheta}\tilde{L}_{n}(\bstheta) - \nabla_{\bstheta}L_{n}(\bstheta)\right\Vert &\leq \sqrt{n} \sup_{\bstheta \in \Theta, \bsz \in \samplespace} \left\| \nabla_{\bstheta} \left[\log( \tilde{f}_{R(n)}(\bsz, \bstheta) ) - \log( f(\bsz, \bstheta) )\right] \right\| \\
   & \leq C_1(f) \sqrt{n} \mathcal{E}( R(n) ) ,
\end{align*}
which converges to zero in probability by Assumption \ref{thm_asynorm_f}(iv).

Using similar arguments, we employ Corollary \ref{cor_concat}(iii) to verify \ref{thm_asynorm_L}(v). We compute
\begin{align*}
\sup_{\bstheta\in\Theta} \left\Vert \nabla_{\bstheta\bstheta}\tilde{L}_{n}(\bstheta) - \nabla_{\bstheta\bstheta}L_{n}(\bstheta)\right\Vert & \leq \sup_{\bstheta \in \Theta, \bsz \in \samplespace} \left\| \nabla_{\bstheta\bstheta} \left[\log( \tilde{f}_{R(n)}(\bsz, \bstheta) ) - \log( f(\bsz, \bstheta) )\right] \right\|\\
&\leq C_2(f) \Bigg( \mathcal{E}( R(n) ) + \sup_{\bstheta \in \Theta, \bsz \in \samplespace}\left\Vert \nabla_{\bstheta}\tilde{f}_{R(n)}(\bsz, \bstheta)-\nabla_{\bstheta}f(\bsz, \bstheta) \right\Vert^2\\
& \quad \quad + \sup_{\bstheta \in \Theta, \bsz \in \samplespace} \left\Vert \nabla_{\bstheta\bstheta}\tilde{f}_{R(n)}(\bsz, \bstheta)-\nabla_{\bstheta\bstheta}f(\bsz, \bstheta) \right\Vert \Bigg) .
\end{align*}

Clearly, as \(n \to \infty\), by the given assumptions, all three summands within the bracket tend to zero in probability. This implies that the whole expression tends to zero, because $C_2(f)$ is independent of $n$.

\end{proof}

\subsection{The level of approximation accuracy} \label{sec_requiredAcc}
In Theorem \ref{thm_asynorm_f} it is required by condition (ii) that \(\plim_{n \to \infty } \sqrt{n}\,\mathcal{E}(R(n)) = 0\), i.e. the approximation error \eqref{eq_error} of the gradient \(\mathcal{E}(R(n))\) needs to decay faster than \(n^{-1/2}\).
Different approximation algorithms behave differently in terms of how the error bounds \(\mathcal{E}(r)\) change with \(r\). We discuss the two most common forms of convergence and their implication on the choice of \(R: \N \to \N\).

\begin{itemize}
\item[(a)] Algebraic convergence: For constants \(c > 0\) and \(s > 0\) it holds
\begin{equation}\label{eq_rate_a}
\mathcal{E}(r) \leq c r^{-s} .
\end{equation}
\item[(b)] Exponential convergence: For constants \(c > 0\) and \(\alpha, \beta > 0\) it holds
\begin{equation}\label{eq_rate_b}
\mathcal{E}(r) \leq c \exp \left( -\alpha r^\beta \right) .
\end{equation}
\end{itemize}
Now, we give general results on how to choose \(R(n)\) depending on the convergence rates of the approximation algorithm for \(\tilde{f}_r\).

\begin{thm} \label{thm_decay_assumption}
 \begin{enumerate}
 \item[(a)] Assume that an algebraic convergence rate \(\mathcal{E}(r) \leq c r^{-s}\) holds for sufficiently large \(r \in \N\) and \(s > 0\). Then, for all \(\gamma > \frac 1 2\) it holds that \(\sqrt{n}\mathcal{E}(R(n)) \convprob 0\) if
 \begin{equation} \label{eqn_decay_a}
  R(n) \geq \left\lceil c^{\frac{1}{s}} n^{\frac{\gamma}{s}} \right\rceil .
 \end{equation}
 \item[(b)] Assume that an exponential convergence rate \(\mathcal{E}(r) \leq c \exp \big( -\alpha r^\beta \big)\) holds for sufficiently large \(r \in \N\) and \(\alpha, \beta > 0\). Then, for all \(\gamma > \frac 1 2\) it holds that \(\sqrt{n}\mathcal{E}(R(n)) \convprob 0\) if
 \begin{equation} \label{eqn_decay_b}
  R(n) \geq \left\lceil \left( \frac{\log c}{\alpha} +  \frac{\gamma}{\alpha} \log n \right)^{\frac{1}{\beta}} \right\rceil .
 \end{equation}
 \end{enumerate}

 \begin{proof}
 For (a), we insert \eqref{eqn_decay_a} into \eqref{eq_rate_a}. This yields
 \begin{align*}
  \sqrt{n} \mathcal{E}(R(n)) \leq c \sqrt{n} \left\lceil c^{\frac{1}{s}} n^{\frac{\gamma}{s}} \right\rceil^{-s} \leq  \sqrt{n} n^{-\gamma} \leq  n^{\frac{1}{2} - \gamma} ,
 \end{align*}
 which, for \(\gamma > 1/2\) tends to zero as \(n \to \infty\).

 For (b), we insert \eqref{eqn_decay_b} into \eqref{eq_rate_b} and obtain
  \begin{align*}
  \sqrt{n} \mathcal{E}(R(n)) \leq c \sqrt{n} \exp \bigg( -\alpha \Big\lceil \Big( \frac{\log c}{\alpha} + \frac{\gamma}{\alpha} \log n \Big)^{\frac{1}{\beta}} \Big\rceil^\beta \bigg) \leq \sqrt{n} \exp \left( - \gamma \log n  \right) = \sqrt{n} \, n^{-\gamma} ,
 \end{align*}
 which also tends to zero as \(n \to \infty\), if \(\gamma > 1/2\).
 \end{proof}
\end{thm}

\section{Specific approximation algorithms} \label{s_rules}
Up to this point, we did not specify any particular approximation. In the following, we examine a very common setting, i.e. a likelihood contribution of the form
\begin{equation} \label{eqn_integral1}
 f(\bsz_{i};\bstheta) = \int_\Omega \varphi(\bsv, \bsz_i, \bstheta) \, \omega(\bsv) \, \rd\bsv ,
\end{equation}
which is a (possibly multivariate) integral of a function \(\varphi: \Omega \times \samplespace \times \Theta \to \R\) over the domain of integration \(\Omega \subset \R^d\) with respect to a given weight function \(\omega: \Omega \to \R_+\).
The integral typically represents the expectation over a nonlinear function \(\varphi\) of a set of random variables with density function \(\omega\).
This class of models includes nonlinear random effects models like \cite{Butler1982} and random coefficients models like the mixed logit model, see for example \cite{McFaddenTrain2000}.

The need for approximation arises because the integral in \eqref{eqn_integral1} often cannot be computed in closed form. To this end, several estimators have been proposed including the method of simulated moments (MSM; \cite{McFadden1989}) and the method of maximum simulated scores (MSS; \cite{HajivassiliouMcFadden1998}). Possibly the most widely used approach is the maximum simulated likelihood (MSL) estimator. It approximates \(f(\bsz_i; \bstheta)\) with a Monte Carlo estimate, i.e. \(\tilde f_r(\bsz_i; \bstheta) = \frac 1 r \sum_{j=1}^r \varphi(\bsv_j, \bsz_i, \bstheta)\), where \((\bsv_j)_{j=1}^r\) is a set of random draws distributed according to the weight function \(\omega\).

Many authors have found that integration rules other than pure Monte-Carlo simulation often perform much better in practice. Examples include Quasi-Monte Carlo rules in \citet{Bhat2001}, Gaussian quadrature in \citet{Butler1982}, or quadrature on sparse grids in \citet{Heiss2008}.

To cover all the mentioned approaches, we consider a general approximation\footnote{Note that both, the weights $w_{j,r}$ and the points $\bsv_{j,r}$ do not depend on the sample $\bsz_i$. This implies that for each likelihood contribution the same approximation algorithm is employed.}
\begin{equation}
 \tilde{f}_r(\bsz_{i};\bstheta) = U_r(\varphi(\cdot, \bsz_i, \bstheta)) := \sum_{j=1}^r w_{j,r} \varphi(\bsv_{j,r}, \bsz_i, \bstheta) 
\end{equation}
to the true integral of \eqref{eqn_integral1}.
This formulation includes Monte-Carlo, where all weights (\(\omega_{j,r}\)) are equal to \(r^{-1}\) and the nodes (\(\bsv_{j,r}\)) are random draws, but also Quasi-Monte Carlo (QMC; like Halton or Sobol sequences), cf. \cite{Halton:1964} or \cite{Sobol:1967}, where the weights are also equal to \(r^{-1}\) but the nodes are deterministic. Moreover, it includes classical (Gaussian) quadrature rules, cf. \cite{rabinowitz}, as well as sparse grids \citep{GerstnerGriebel:1998}. We will come back to specific algorithms below.

Remember that for continuous \(\varphi(\cdot, \bsz, \bstheta)\) our likelihood contributions \(f(\bsz, \bstheta)\) and their derivatives now have the following general form 
\begin{equation} \label{eqn_integral1rep}
 \begin{aligned}
 f(\bsz;\bstheta) &= \int_\Omega \varphi(\bsv, \bsz, \bstheta) \, \omega(\bsv) \, \rd\bsv \\
 \nabla_{\bstheta}f(\bsz;\bstheta) &= \int_\Omega \nabla_{\bstheta}\varphi(\bsv, \bsz, \bstheta) \, \omega(\bsv) \, \rd\bsv \\
 \nabla_{\bstheta\bstheta}f(\bsz;\bstheta) &= \int_\Omega \nabla_{\bstheta\bstheta}\varphi(\bsv, \bsz, \bstheta) \, \omega(\bsv) \, \rd\bsv .
 \end{aligned} 
\end{equation}

Their associated approximations \(\tilde{f}_r(\bsz, \bstheta)\) stem from some cubature rule \(U_r\) of the form
\begin{equation} \label{eqn_cub1rep}
 \begin{aligned}
  \tilde{f}_r(\bsz;\bstheta) &= U_r(\varphi(\cdot, \bsz, \bstheta)) &&= \sum_{j=1}^r w_{j,r} \varphi(\bsv_{j,r}, \bsz, \bstheta)\\
 \nabla_{\bstheta}\tilde{f}_r(\bsz;\bstheta) &= U_r(\nabla_{\bstheta}\varphi(\cdot, \bsz, \bstheta)) &&= \sum_{j=1}^r w_{j,r} \nabla_{\bstheta}\varphi(\bsv_{j,r}, \bsz, \bstheta)\\
 \nabla_{\bstheta\bstheta}\tilde{f}_r(\bsz;\bstheta) &= U_r(\nabla_{\bstheta\bstheta}\varphi(\cdot, \bsz, \bstheta)) &&= \sum_{j=1}^r w_{j,r} \nabla_{\bstheta\bstheta}\varphi(\bsv_{j,r}, \bsz, \bstheta) .
 \end{aligned} 
\end{equation}

The goal in this section will be to impose conditions on \(\varphi\) as well as on the integration rule \(U_r\) such that the assumptions of Theorems \ref{thm_consist_f} and \ref{thm_asynorm_f} hold. To this end, the following notation for the partial derivative of a sufficiently smooth function $g$ will be helpful.
\begin{equation} \label{eqn_partial_deriv}
 D_{\bsalpha}^{(\bstheta)} g(\bstheta) = \frac{\partial^{|\bsalpha|}}{\prod_{j=1}^p \theta_j^{\alpha_j}} g(\bstheta) ,
\end{equation}
where \(\bsalpha = (\alpha_1,\ldots, \alpha_p)\) is a multi-index that contains the order of the partial derivatives for different coordinate directions \(\theta_1,\ldots,\theta_p\).

In order to bound the magnitude of \(\sup_{\bsz\in\samplespace,\bstheta\in\Theta}\left|\tilde{f}_r(\bsz, \bstheta) - f(\bsz, \bstheta)\right|\), as well as \(\sup_{\bsz\in\samplespace,\bstheta\in\Theta} \|\nabla_{\bstheta}\tilde{f}_r(\bsz;\bstheta) - \nabla_{\bstheta}f(\bsz;\bstheta) \|\) and \(\sup_{\bsz\in\samplespace,\bstheta\in\Theta} \|\nabla_{\bstheta\bstheta}\tilde{f}_r(\bsz;\bstheta) - \nabla_{\bstheta\bstheta}f(\bsz;\bstheta) \|\) defined by \eqref{eqn_integral1rep} and \eqref{eqn_cub1rep}, we resort to
\begin{equation} \label{eqn_barE}
 \begin{aligned}
 \bar{\mathcal{E}}_k(r) & := \max_{|\bsalpha|_1 \leq k} \sup_{\bsz\in\samplespace,\bstheta\in\Theta}\left| D_{\bsalpha}^{(\bstheta)} f(\bsz, \bstheta) - D_{\bsalpha}^{(\bstheta)} \tilde{f}_r(\bsz, \bstheta) \right| \\
 & = \max_{|\bsalpha|_1 \leq k} \sup_{\bsz\in\samplespace,\bstheta\in\Theta}\left| \int_\Omega D_{\bsalpha}^{(\bstheta)} \varphi(\bsv, \bsz, \bstheta) \, \omega(\bsv) \, \rd\bsv - \sum_{j=1}^r w_{j,r} D_{\bsalpha}^{(\bstheta)} \varphi(\bsv_{j,r}, \bsz, \bstheta) \right|
 \end{aligned}
\end{equation}
for \(k=0,1,2\). Since all matrix- and vector-norms can be bounded by the \(\ell_\infty\) norm it holds \(\mathcal{E}(r) \leq 2 p \bar{\mathcal{E}}_1(r)\) and hence it is sufficient (and convenient) to work with the quantity defined in \eqref{eqn_barE}.

Now, \(\lim_{r \to \infty} \bar{\mathcal{E}}_k(r) = 0\) basically implies that the quadrature rules \((U_r)_{r \in \N}\) in \eqref{eqn_cub1rep} converge for \(\varphi(\cdot, \bsz, \bstheta)\) and all partial derivatives with respect to \(\bstheta\) up to order \(k\).

\subsection{Consistency and asymptotic normality}
First, we deal with consistency, which follows as a simple application of Theorem \ref{thm_asynorm_L} to the specific setting where \(f\) has the form \eqref{eqn_integral1rep} and \(\tilde{f}_r\) has the form \eqref{eqn_cub1rep}.

\begin{corollary}\label{thm_consist_phi}
Assume that the conditions of Lemma \ref{extremum_estimator-consist} hold and that
(i) there exists \(\bar{\delta} > 0\) such that for all values of \(\bsz \in \samplespace \) and \(\bstheta \in \Theta\) it holds \(f(\bsz;\bstheta) > \bar{\delta} > 0\);
(ii) the function \(\varphi(\bsv, \bsz, \bstheta)\) is continuous in \(\bstheta\in\Theta\);
(iii) \(\plim_{r \to \infty} \bar{\mathcal{E}}_0(r) = 0\);
(iv) \(R(n)\) is monotonically increasing in \(n\).
Then, \(\hat{\bstheta}_{MAL}\) is a consistent estimator of \(\bstheta_0\), i.e.
\begin{equation}
\hat{\bstheta}_{MAL}\convprob \bstheta_{0}.
\end{equation}
\end{corollary}

Condition \ref{thm_consist_phi}(iii) states that the cubature rule \(U_r\) converges for \(\varphi\), which is a rather mild assumption since no requirements are made regarding the speed of convergence.

Next, we consider the asymptotic distribution of \(\hat{\bstheta}_{MAL}\). Again, \(f\) has the form \eqref{eqn_integral1rep} and \(\tilde{f}_r\) has the form \eqref{eqn_cub1rep}.
The critical part of the properties of \(\hat\bstheta_{MAL}\) is the worst-case approximation error, i.e. \eqref{eqn_barE} for \(k=1\), which not only has to go to zero, but also has to vanish at a certain speed.

\begin{corollary}\label{thm_asynorm_phi}
Assume that the assumptions of Lemma \ref{thm_ML_distr} and Corollary \ref{thm_consist_phi} hold.
Also assume that
(i) for all \(\bsv \in \Omega\) and all \(\bsz \in \samplespace \) it holds that \(\varphi(\bsv, \bsz, \cdot)\) is twice continuously differentiable;
(ii) for all \(\bsz \in \samplespace\) and all \(\bstheta \in \Theta\) all terms in \eqref{eqn_integral1rep} remain bounded.
(iv) \(\plim_{r \to \infty} \bar{\mathcal{E}}_2(r) = 0\);
(v) \(R(n)\) rises fast enough with \(n\) to ensure \(\sqrt{n}\bar{\mathcal{E}}_1 \big( R(n) \big) \convprob 0\).
Then
\[
\sqrt{n}\left(\hat{\bstheta}_{MAL}-\bstheta_{0}\right)\overset{d}{\rightarrow}N(\boldzero,\mathcal{I}^{-1}).
\]
\end{corollary}

\begin{remark}
Note that at no point the positivity of the integration weights is assumed. Only (sufficiently fast) convergence of \(U_r\) for all partial derivatives of \(\varphi\) up to order \(2\) is required.
\end{remark}

\begin{remark}
 Note that the results obtained in the last three Sections are of asymptotical nature, i.e. they hold for \(n \to \infty\). In practice, \(n\) and also \(r = R(n)\) need to be finite numbers, where \(R(n)\) has to be chosen such that \(R(n) \geq r_0(\delta)\) because otherwise \(\tilde{f}_{R(n)} > \delta\) might
 not hold true. Especially for \(\tilde{f}_{R(n)} \leq 0\) the whole approach does not work anymore because of the logarithm in the log-likelihood  contributions.
\end{remark}

\subsection{Examples for integration rules} \label{s_int_rules}
In this section we will turn from the previous abstract framework to more explicit approximation algorithms and likelihood constructions. We will discuss several choices for the approximation of the integrals \eqref{eqn_integral1rep} and their resulting complexities. To this end, remember that the likelihood function consists of \(n\) terms where each one requires the numerical approximation of an integral with \(r = R(n)\) evaluations of \(\varphi\). Therefore, the total complexity for one evaluation of the approximated likelihood function is \(n \cdot R(n)\), with the requirements on \(R(n)\)  given in Theorem \ref{thm_decay_assumption}. Several concrete examples are provided in Table \ref{table_cost}.

\begin{table}[b]
\centering

 \begin{tabular}{|c|c|c|}
 \hline
 Setting & Error bound & Total cost \(n \cdot R(n)\)\\
 \hline
 Monte Carlo & \(\cO(r^{-1/2})\) & \(\cO(n^2)\) \\
 Quasi -- Monte Carlo & \(\cO(r^{-1+ \varepsilon})\) & \(\cO(n^{3/2 + \varepsilon})\) \\
 Gaussian quadrature, smoothness \(k=2\) & \(\cO(r^{-2})\) & \(\cO(n^{5/4})\) \\
 Gaussian quadrature, analytic function & \(\cO(e^{-r})\) & \(\cO(n \log(n))\) \\
 Sparse grid, bounded domain & \(\cO(r^{-k + \varepsilon})\) & \(\cO(n^{1 + \frac{1}{2k} + \varepsilon})\) \\
 Sparse grid, unbounded domain & \(\cO(r^{-k/2 + \varepsilon})\) & \(\cO(n^{1 + \frac{1}{k} + \varepsilon})\) \\
 \hline
 \end{tabular}

 \caption{Total number of integrand evaluations for MALE with different integration schemes.} \label{table_cost}
\end{table}

\subsubsection{Monte Carlo simulation}
In order verify the consistency of the Monte Carlo simulation (MC) approach it only is required that \(\varphi(\cdot, \bsz, \bstheta)\) is uniformly bounded in  \(L^2(\Omega, \omega)\) for all data \(\bsz \in \samplespace\) and all parameters \(\bstheta \in \Theta\). This means that \(\sup_{\bsz, \bstheta} \int_\Omega \varphi(\bsv, \bsz, \bstheta)^2 \, \omega(\bsv) \rd \bsv < \infty\).
Then, the choice \(w_{j,r} = 1/r\) and independent and identical draws \(\bsv_{j,r}\) from the probability distribution induced by \(\omega\) yield a quadrature rule that converges both, in expectation and with high probability at a rate of \(\mathcal{O}(r^{-1/2})\).
By Corollary \ref{thm_consist_phi} the resulting estimator is consistent.

For asymptotical normality, also the partial derivatives of \(\varphi\) with respect to \(\bstheta\) must be square integrable, i.e.
\[
	\sup_{\bsz, \bstheta} \int_\Omega \left( D^{(\bstheta)}_\bsalpha \varphi(\bsv, \bsz, \bstheta) \right)^2 \, \omega(\bsv) \rd \bsv < \infty \quad \text{ for all } |\bsalpha|_1 \leq 2.
\]
Then, with high probability and also in expectation, \(\bar{\mathcal{E}}_1 \leq c r^{-1/2}\) and \(\bar{\mathcal{E}}_1\) converges to zero at a rate of \(\mathcal{O}(r^{-1/2})\). It follows by Corollary \ref{thm_asynorm_phi} and Theorem \ref{thm_decay_assumption} that the Monte Carlo based MAL estimator is asymptotically normal if the number of integration points \(r\) increases at least linearly in the number of data samples \(n\), i.e. \(r = R(n) = n^\beta\) with \(\beta > 1\).

\subsubsection{Quasi -- Monte Carlo integration}
Another approach to multivariate integration on the \(d\)-dimensional unit cube \(\Omega = (0,1)^d\) is the quasi--Monte Carlo method. Similar to Monte Carlo it has all weights \(w_j = r^{-1}\), but the points are not drawn randomly, but determined deterministically by certain number-theoretic considerations. Quasi--Monte Carlo methods exist in differenct variants, depending on the specific selection of points. Examples are Halton points, cf. \cite{Halton:1964}, lattice rules, cf. \cite{SloanJoe:1994} or Sobol points, cf. \cite{Sobol:1967}. A rather recent development
are higher-order QMC points, cf. \cite{DiPi10,Hinrichs.Markhasin.Oettershagen.ea:2016}.

For integrands that possess bounded variation in the sense of Hardy and Krause, cf. \cite{Niederreiter:1992}, most QMC constructions achieve a convergence rate of order \(\cO(r^{-1} \log(r)^{d-1})\). Asymptotically, this bound behaves like \(\cO(r^{-1 + \varepsilon})\), where the \(\varepsilon > 0\) asymptotically suppresses the \(d\)-dependent power of \(\log r\).

Details on the definition of the Hardy-Krause variation can be found in \cite{DiPi10,Niederreiter:1992,owen2005multidimensional}. We just note at this point that bounded variation is a much stronger assumption than bounded variance (as it is required for Monte Carlo integration) since for bounded Hardy-Krause variation also the mixed derivatives need to exist and to be bounded, i.e.
\begin{equation} \label{eqn_vitali_variation}
 \int_{[0,1]^d} \frac{\partial^d}{\prod_{j=1}^d \partial v_j} g(\bsv) \, \rd \bsv < \infty \quad \text{ for all } (\bsz, \bstheta) \in \samplespace \times \Theta .
\end{equation}

If this property is fulfilled for \(g(\bsv) = \varphi(\bsv, \bsz, \bstheta)\), the resulting estimator will be consistent. If in addition also \eqref{eqn_vitali_variation} is fulfilled for all partial derivatives of \(\varphi\) up to order \(2\) with respect to \(\bstheta\), i.e. \(g(\bsv) = D^{(\bstheta)}_\bsalpha \varphi(\bsv, \bsz, \bstheta), |\bsalpha|_1 \leq 2\) then both, \(\bar{\mathcal{E}}_1(r)\) and \(\bar{\mathcal{E}}_2(r)\) decay at a rate of \(\cO(r^{-1 + \varepsilon})\) and Corollary \ref{thm_asynorm_phi} and Theorem \ref{thm_decay_assumption} yield that the number of QMC integration points needs to increase as \(r = R(n) = n^{\beta}\) with any \(\beta > 1/2\).

We remark that in some cases the condition of bounded Hardy-Krause variation can be relaxed to deal with integrands that possess mild singularities, cf. \cite{Owen}.

\subsubsection{Gaussian quadrature and related methods}
The classical approach to numerical integration aims at the so-called degree of exactness. This means that the quadrature rule is constructed such that it integrates polynomials up to a certain degree \(D\) exactly, i.e. the points \(v_{1,r},\ldots, v_{r,r}\) and weights \(w_{1,r},\ldots, w_{r,r}\) are determined such that
\begin{equation}
 \sum_{j=1}^r w_{j,r} v_{j,r}^k = \int_\Omega v^k \, \omega(v) \rd v \quad \text{ for all } k =0, 1, \ldots, D .
\end{equation}
The quadrature rule with the best possible degree of exactness \(2r-1\) is Gaussian quadrature. Depending on the smoothness of the integrand, Gaussian quadrature yields algebraic convergence rates (for integrands with finite smoothness) or (sub-)exponential convergence rates for infinitely differentiable or analytic integrands, cf. \cite{rabinowitz}. 
We will discuss two examples for Gaussian quadrature in more detail below.
Note that it is also possible to extend the Gaussian approach to non-polynomial basis functions, e.g. to deal with certain boundary singularities that occur with the GHK sampling approach, cf. \cite{Griebel.Oettershagen:2014}.
Moreover, there are further approaches that yield quadrature rules with a polynomial degree of exactness that involve nested point sets. Examples are Gauss-Patterson quadrature rules, cf. \cite{Patterson1967}, Clenshaw-Curtis quadrature, cf. \cite{Imhof1963} or Leja points, cf. \cite{jantsch2016lebesgue,Griebel.Oettershagen:2016}.

Besides favorable convergence rates, Gaussian quadrature rules have the property that the Stone-Weierstrass Theorem and its variations ensure their convergence for any continuous integrand. Therefore, by Theorem \ref{thm_consist_phi} maximum approximated likelihood estimators based on Gaussian quadrature rules and other stable polynomial-based quadrature rules usually are consistent if \(\varphi(\cdot, \bsz, \bstheta)\) is continuous for all \((\bsz, \bstheta) \in \samplespace \times \Theta\).

\paragraph{Gauss-Legendre quadrature} If the weight function is constant, i.e. \(\omega(x) =1\) and \(\Omega\) is bounded, e.g. \(\Omega = [0,1]\), Gauss-Legendre quadrature achieves the best possible degree of polynomial exactness. As a consequence, it achieves exponential convergence rates of type \(\cO(e^{-\alpha r}), \alpha > 0\) for integrands that are analytic in certain ellipses or circles that enclose the domain of integration \(\Omega\), cf. \cite{rabinowitz}. However, if the integrand is only \(k < \infty\) times differentiable, the convergence rate deteriorates to the order \(\cO(r^{-k})\).

\paragraph{Gauss-Hermite quadrature} If the weight function corresponds to a standard normal density, e.g. \(\omega(x) = \frac{1}{\sqrt{2 \pi}} e^{-x^2/2}\) and \(\Omega = \R\), then Gauss-Hermite quadrature is optimal with respect to polynomial exactness. The analysis of Gaussian quadrature on unbounded domains \(\Omega\) is more complicated than in the case of bounded \(\Omega\). Yet, there are a number of useful results available: For integrands with \(k\) continuous and integrable derivatives, the error of Gauss-Hermite quadrature can be bounded by \(\cO(r^{-k/2})\), cf. \cite{Smith2018} or \cite{Mastroianni94}.
To be more precise, if there exists a constant \(c > 0\) such that
\begin{equation} \label{eqn_hermite_condition_1d}
 g^{(k)}(v) \leq c \frac{e^{\frac{v^2}{2}}}{\sqrt{1+v^2}} \quad \text{ for all } v \in \R
\end{equation}
then it holds
\begin{equation}
 \left| \int_\R g(v) \frac{e^{-\frac{v^2}{2}}}{\sqrt{2\pi}} \, \rd v - \sum_{j=1}^r w_{j,r} g(v_{j,r}) \right| \leq C(k,g) r^{-k/2} ,
\end{equation}
where the constant \( C(k,g) > 0\) may depend on \(k\) and the integrand \(g\), but not on \(r\). Basically, condition \eqref{eqn_hermite_condition_1d} bounds the growth of the $k$-th derivative along the real axis.
Moreover, for integrands that are analytic in an infinite complex strip that contains the real axis \(\R\), sub-exponential convergence rates of type \(\cO(e^{-\alpha n^\beta})\) are shown in \cite{boyd1984}.
Finally, for certain classes of entire functions, even exponential convergence \(\cO(e^{-\alpha r})\) is possible, cf. \cite{Kuo2011} or \cite{Irrgeher2015}.

\subsubsection{Sparse grid cubature}
For multivariate integration problems it is possible to use sparse grids, cf. \cite{GerstnerGriebel:1998,GerstnerGriebel:2003,Griebel.Oettershagen:2016,Novak1996} to turn a univariate quadrature rule to a multivariate integration method.

For integrands with bounded mixed derivatives up to order \(k\) on bounded domains a classical result is the convergence rate of order \(\cO(r^{-k} (\log r)^{(d-1)(r+1)})\), cf. \cite{GerstnerGriebel:1998} and \cite{Novak1996}. In the case of sparse grid Gauss-Hermite quadrature on the unbounded domain \(\R\) it holds, cf. \cite{Zhang2013},
\begin{equation}
 \left| \int_{\R^d} g(\bsv) \frac{e^{-\frac{\bsv^t \bsv}{2}}}{(2\pi)^{d/2}} \, \rd \bsv - \sum_{j=1}^r w_{j,r} g(\bsv_{j,r}) \right| \leq C(k,g) r^{-k/2} (\log r)^{(d-1)(k/2+1)} ,
\end{equation}
which asymptotically behaves like \(\cO(r^{-k/2 + \varepsilon})\).
Here, as before \(C(k,g)\) depends on \(k\) and \(g\) and the integrand \(g\) must have continuous partial derivatives up to order \(k\) which fulfill
\begin{equation} \label{eqn_hermite_condition_mv}
 \frac{\partial^{k d}}{\prod_{l=1}^d \partial v_l^k} g(\bsv) \leq c \frac{ e^{\frac{\bsv^t \bsv}{2}}}{\prod_{l=1}^d \sqrt{1+v_l^2}} \quad \text{ for all } \bsv \in \R^d ,
\end{equation}
i.e. the mixed derivative of order \(k\) is continuous and does not grow too fast as \(\bsv\) approaches infinity.
For integrands that are infinitely often differentiable, \(k \in \N\) can be chosen as large as desired, but then also the constant \(C(k,g)\) might become arbitrary large.

\section{Examples} \label{s_examples}

\paragraph{Example Ia}
In the logit model, one needs to compute integrals of the form
\begin{equation} \label{eqn_logit_integral1}
 \int_\R \frac{1}{1 + e^{- z_i x}} \, \frac{e^{-\frac{(x-\mu)^2}{2 \sigma^2}}}{\sqrt{2 \pi \sigma^2}} \rd x .
\end{equation}
Since the parameter vector \(\bstheta = (\theta_1,\theta_2) = (\mu, \sigma)\) shall only contribute to \(\varphi\) in \eqref{eqn_integral1} we reformulate the integral \eqref{eqn_logit_integral1} using the change of variable \(v = (x-\mu)/\sigma\) to obtain
\begin{align*}
 \int_\R \frac{1}{1 + e^{- z_i x}} \, \frac{e^{-\frac{(x-\mu)^2}{2 \sigma^2}}}{\sqrt{2 \pi \sigma^2}} \rd x & = \int_\R \frac{1}{1 + e^{ - z_i ( {\sigma} v + \mu )}} \frac{e^{-\frac{v^2}{2}}}{\sqrt{2 \pi}} \rd v .
\end{align*}
Hence, in the notation of \eqref{eqn_integral1} we have with \(\bstheta = (\theta_1,\theta_2) := (\mu, \sigma)\)
\begin{equation}
 \Omega = \R, \quad \omega(v) = \frac{1}{\sqrt{2 \pi}} e^{-\frac{v^2}{2}}, \quad \varphi(v, z_i, \bstheta) = \frac{1}{1 + e^{- z_i ( {\theta_2} {v} + \theta_1 )}} .
\end{equation}
In order to determine the convergence rate of Gauss-Hermite quadrature, we have to analyze the function \(v \mapsto \varphi(v, z_i, \bstheta)\), as well as its derivatives with respect to \(\bstheta\), i.e. the functions \(v \mapsto D_\bsalpha^{(\bstheta)} \varphi(v, z_i, \bstheta)\) for \(|\bsalpha|_1 \leq 2\).
To this end, we use Lemma \ref{cor_logit_derivative} from the Appendix, which ensures that \eqref{eqn_hermite_condition_1d} is satisfied for all \(k \in \N\), i.e.
\begin{equation}
 \frac{\rd^k}{\rd v^k} D_\bsalpha^{(\bstheta)} \varphi(v, z_i, \bstheta) \leq c(k, \bsalpha) \frac{e^{\frac{v^2}{2}}}{\sqrt{1+v^2}} \quad \text{ for all } v \in \R \text{ and all } \bsalpha \in \N_0^2 .
\end{equation}
Therefore, we can expect arbitrary large algebraic convergence of order \(\cO(r^{-k})\).

\paragraph{Example Ib}
In the multivariate random coefficients logit model one needs to compute integrals of the form
\begin{equation} \label{eqn_logit_integral2}
 \int_{\R^d} \frac{1}{1 + e^{-\bsz_i \cdot \bsv}} \, \frac{ \exp\big(-\frac{1}{2} (\bsv-\bsmu)^t \Sigma^{-1} (\bsv-\bsmu) \big) }{\sqrt{(2 \pi)^d \det \Sigma}} \rd \bsv .
\end{equation}

We use the variable transformation
\[
 \bsu = \bsC^{-1}(\bsv - \bsmu),
\]
where \(\bsC\) denotes the Cholesky factorization of \(\Sigma\), i.e. \(\Sigma = \bsC \bsC^t\). Then, for any \(f: \R^d \to \R\) we obtain
\begin{align*}
 & \quad \int_{\R^d} f(\bsv) \, \frac{ \exp\big(-\frac{1}{2} (\bsv-\bsmu)^t \Sigma^{-1} (\bsv-\bsmu) \big) }{\sqrt{(2 \pi)^d \det \Sigma}} \rd \bsv  \\
 & = \int_{\R^d} f(\bsC (\bsu + \bsmu) ) \, \frac{ \exp\big(-\frac{1}{2} \bsu^t \bsC^t \Sigma^{-1} \bsC^t \bsu \big) }{\sqrt{(2 \pi)^d \det \Sigma}} \det \bsC \rd \bsu \\
 & = \int_{\R^d} f(\bsC (\bsu + \bsmu) ) \, \frac{ \exp\big(-\frac{1}{2} \bsu^t \bsu \big) }{\sqrt{(2 \pi)^d }} \rd \bsu .
\end{align*}

Hence, in the notation of \eqref{eqn_integral1rep} we have with \(\bstheta = (\theta_1,\theta_2, \ldots, \theta_p) := (\bsmu, \bsC)\), where \(p = d + d(d+1)/2\), that
\begin{equation}
 \Omega = \R^d, \quad \omega(\bsv) = \frac{e^{-\frac{\bsv^t \bsv}{2}}}{\sqrt{(2 \pi)^d}}, \quad \varphi(\bsv, \bsz_i, \bstheta) = \frac{1}{1 + e^{- \bsz_i \cdot \bsC (\bsv + \bsmu)}} .
\end{equation}

Again, we can use Lemma \ref{cor_logit_derivative} from the Appendix to ensure that \eqref{eqn_hermite_condition_mv} is fulfilled for all \(k \in \N\), i.e.
\begin{equation}
 \frac{\partial^{k d}}{\prod_{l=1}^d \partial v_l^k} D_\bsalpha^{(\bstheta)} \varphi(\bsv, \bsz_i, \bstheta) \leq c(k, \bsalpha) \frac{ e^{\frac{\bsv^t \bsv}{2}}}{\prod_{l=1}^d \sqrt{1+v_l^2}} \quad \text{ for all } \bsv \in \R^d  \text{ and all } \bsalpha \in \N_0^p .
\end{equation}
Therefore, for any \(k \in \N\) we can (asymptotically) expect algebraic convergence of order \(\cO(r^{-k/2} (\log n)^{(d-1)(k/2+1)})\), which asymptotically behaves like \(\cO(r^{-k/2 + \varepsilon})\).

\paragraph{Example 2}
In the Butler-Moffitt model one has to solve univariate integrals of the form
\[
	f(\bsz_i, \bstheta) = \int_\R \prod_{t=1}^T \Phi( z_{i,t} \beta + \sigma v) \frac{e^{-\frac{v^2}{2}}}{\sqrt{2 \pi}} \, \rd v ,
\]
where \(\Phi: \R \to (0,1)\) denotes the cumulative distribution functions of the standard normal distribution.

Hence, in the notation of \eqref{eqn_integral1} we have with \(\bstheta = (\theta_1,\theta_2) := (\sigma, \beta)\)
\begin{equation}
 \Omega = \R, \quad \omega(v) = \frac{1}{\sqrt{2 \pi}} e^{-\frac{v^2}{2}}, \quad \varphi(v, \bsz_i, \bstheta) = \prod_{t=1}^T \Phi( z_{i,t} \beta + \sigma v) .
\end{equation}
In order to determine the convergence rate of Gauss-Hermite quadrature, we have to analyze the function \(v \mapsto \varphi(v, \bsz_i, \bstheta)\), as well as its derivatives with respect to \(\bstheta\), i.e. the functions \(v \mapsto D_\bsalpha^{(\bstheta)} \varphi(v, \bsz_i, \bstheta)\) for \(|\bsalpha|_1 \leq 2\).
To this end, we note that \(\Phi\) is bounded and all of its derivatives are also bounded. Therefore, by the general Leibnitz / product rule, \eqref{eqn_hermite_condition_1d} is fulfilled for all \(k \in \N\), i.e.
\begin{equation}
 \frac{\rd^k}{\rd v^k} D_\bsalpha^{(\bstheta)} \varphi(v, z_i, \bstheta) \leq c(k, \bsalpha) \frac{e^{\frac{v^2}{2}}}{\sqrt{1+v^2}} \quad \text{ for all } v \in \R \text{ and all } \bsalpha \in \N_0^2 .
\end{equation}
Thus, we can expect arbitrary large algebraic convergence of order \(\cO(r^{-k})\).

\begin{remark}
Note that an \emph{arbitrary large algebraic rate of convergence} usually amounts to an (sub-)exponential convergence of type $\cO(\exp(-\alpha r^\beta))$. However, to determine the exact rate of decay, i.e. $\alpha$ and $\beta$ or even to check the required assumptions, usually requires advanced techniques from complex analysis, cf. e.g. \cite{rabinowitz} or \cite{boyd1984}, which are beyond the scope of this article. Therefore, at this point, we will stay with algebraic convergence rates for any $k \in \mathbb{N}$ and their much simpler to check prerequisites.
\end{remark}

\section{Pre-Asymptotics} \label{sec_pre_asymp}
So far we have discussed the asymptotic behavior of the maximum approximated likelihood estimator where $n$ and $r$ approach infinity. In the following simulation study we focus on the illustration of the theoretical findings as well as practical relevant insights into the pre-asymptotic behavior. First we demonstrate the relative performance of common approximation methods and link the result to our theory.
Second we consider the consequences caused by the approximation error for the estimation accuracy.

\subsection*{Relative Performance of Approximation Methods}
To ensure consistency in the context of Corollary 10 an approximation method is required for which the approximation error disappears if $r \rightarrow \infty$. For all reasonable approaches, including Monte-Carlo methods, quasi Monte-Carlo and numerical quadrature this assumption is fulfilled. However, for a finite $r$ the approximation accuracy can differ dramatically.

For our illustration we analyze a simplified random coefficient regression model of the form

\[
y_{i}=x_{i}\beta_{i}+\varepsilon_{i}\quad\text{with \ensuremath{\varepsilon_{i}\sim N(0,1)\quad \beta_{i}\sim N(\bar{\beta},1)}}.
\]

The likelihood contribution is
\begin{equation}\label{model}
    f\left(z_i=(y,x), \theta = \bar\beta\right)=\int_{\mathbb{R}}g(y-x\beta)\times g(\beta-\bar{\beta})\,\text{d}\beta,
\end{equation}

where $g(\cdot)$ denotes the standard normal density function with $g(x)=\frac{1}{\sqrt{2\pi}}e^{-x^{2}/2}$. The likelihood contribution is evaluated at $\theta = 0$ using

\[
\tilde{f}\left(z_i, \theta \right) =\sum_{j=1}^{r}w_{j,r}g(y-xv_{j,r}),
\]
where $v_{j,r}$ denotes the draws or nodes regarding the standard
normal distribution and $w_{j,r}$ the corresponding weights.
For every run in the experiment we randomly generate $x_i\sim N(0,1)$, $\beta_i\sim N(0,1)$ and $\varepsilon_i \sim N(0,1)$. Based on the generated data we compute the approximated likelihood contribution.

We apply three approximation methods, i.e. (1) ordinary Monte-Carlo sampling using pseudo-random draws, (2) quasi Monte-Carlo sampling proposed by Halton (1964) and (3) Gauss-Hermite quadrature\footnote{We also tested quasi Monte-Carlo draws following \cite{Sobol:1967} and, representative for variance reduction techniques, we use Modified Latin Hypercube Sampling (MLHS) \citep*[see][]{Hess2006} as well as Gauss-Legendre quadrature and the midpoint-rule. The general result is that Monte-Carlo performs worse, quasi Monte-Carlo and MLHS perform roughly the same and Gaussian-Quadrature performs best.}. As reference $f_{\text{ref}} \simeq f$ we compute the approximation $\tilde{f}(z_i, \theta = 0)$ with a very fine Gauss-Hermite quadrature using 100 grid-points.

For an increasing sequence of $r$ we repeat the experiment $m=5.000$ times and compute for each repetition the approximation error ($\tilde{f}-f_{\text{ref}} $). We aggregate the results in Figure \ref{Conv_fRC} using the "maximum absolute error" (left panel) and the RMSE$=\sqrt{\frac{1}{m}\sum_{i=1}^{m}(\tilde{f}-f_{\text{ref}})^2}$ (right panel). The maximum absolute error complies with the worst case error in assumption (iii) of Theorem \ref{thm_consist_f}. The RMSE is  a more common measure of convergence used here for the comparison of the approximation methods. 

\begin{figure}[H]
\begin{centering}
\includegraphics[width=0.9\textwidth]{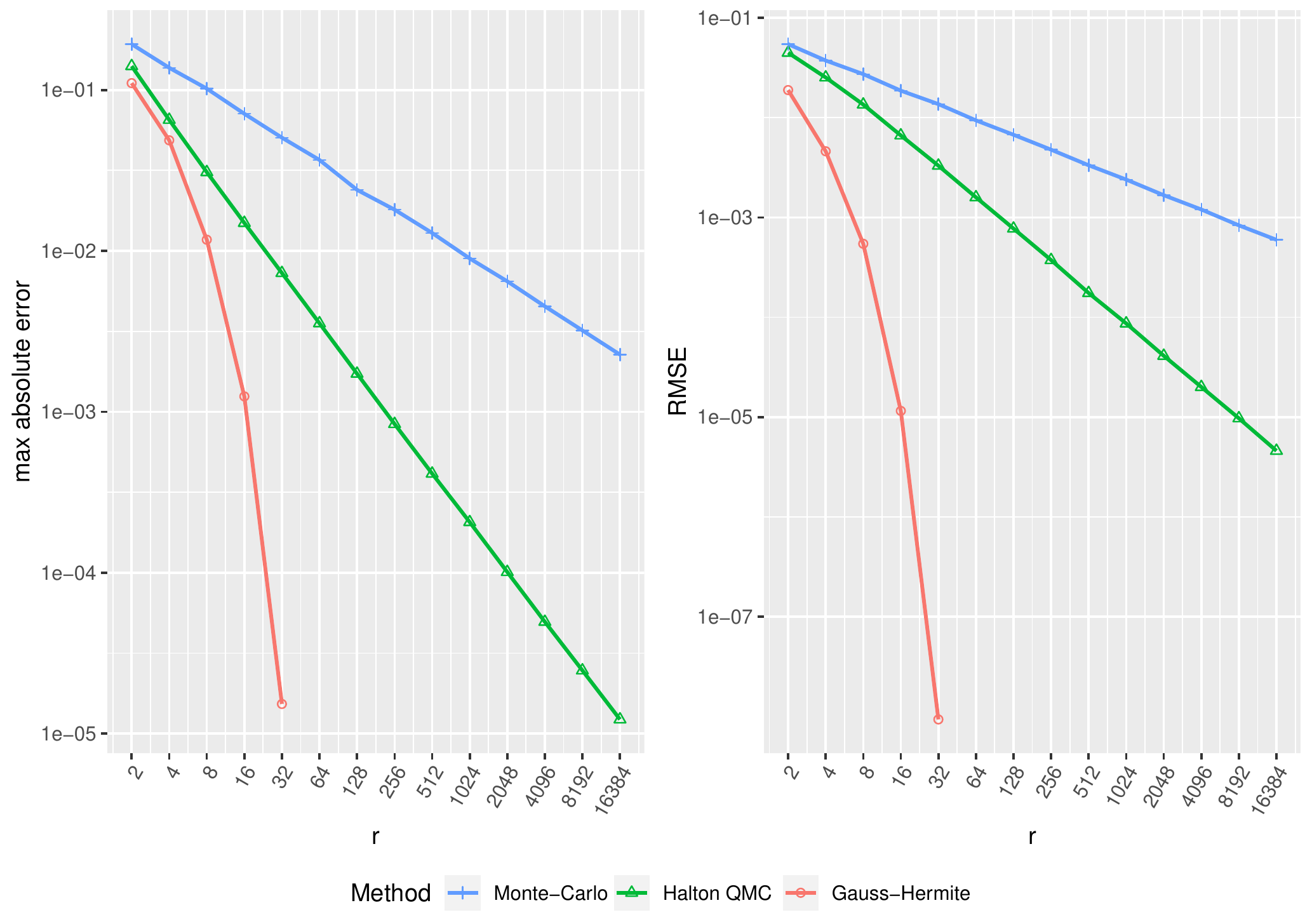}
\par\end{centering}
\caption{\label{Conv_fRC}Convergence behavior using different approximation methods for a smooth function $f\left(z_i=(y,x), \theta = \bar\beta\right)=\int_{\mathbb{R}}g(y-x\beta)\times g(\beta-\bar{\beta})\,\text{d}\beta$}
\end{figure}

As shown in Figure \ref{Conv_fRC} the three methods are converging in the worst case error and in the average error, respectively. But the methods differ in their convergence rates. The function $f(\cdot)$ used here is smooth and therefore we observe the corresponding theoretical error bounds of the different employed methods (see Table \ref{table_cost}). As expected, Gauss-quadrature outperforms the other methods and obtains even an exponential convergence rate. Thus, to reach at least the same accuracy as Monte-Carlo with 16.384 draws, just 128 Halton-draws and even only 16 nodes for the Gauss-Hermite quadrature are needed.

As stated in section \ref{s_int_rules} the convergence behavior and the relative performance is related to the smoothness of the approximated function. In general there is no absolute dominant approximation method but, based on the properties of the likelihood function, one method could perform better than another, provided that more smoothness is present, compare Table 1.

To illustrate this result we also analyze the Accept-Reject-Sampling (ARS) solution for the standard-normal-cdf (see equation (\ref{model_ARS})). This model has a discontinuity at $z=x$ and is therefore not smooth in terms of section \ref{s_int_rules}. We consider

\begin{equation}\label{model_ARS}
    f\left(z\right)=\int_{\mathbb{R}} 1(x \leq z) \cdot g(x) \,\text{d}x,
\end{equation}

where $1(\cdot)$ denotes the indicator function, which is one if the expression in brackets is true and zero otherwise. This function is approximated using

\[
\tilde{f}\left(z\right) =\sum_{j=1}^{r}w_{j,r}1(v_{j,r} \leq z) ).
\]

For our simulation study we randomly draw 5.000 z-values from a standard-normal-distribution for each accuracy level $r$. We apply Monte-Carlo, Halton draws and Gauss-Legendre quadrature. As reference solution we take the error-function representation $f_{\text{ref}} = \frac {1}{2} \left[1+\operatorname {erf} \left(\frac {z }{\sqrt {2}}\right)\right]$. The results are summarized in Figure \ref{Conv_fARS}.
 
\begin{figure}[H]
\begin{centering}
\includegraphics[width=0.9\textwidth]{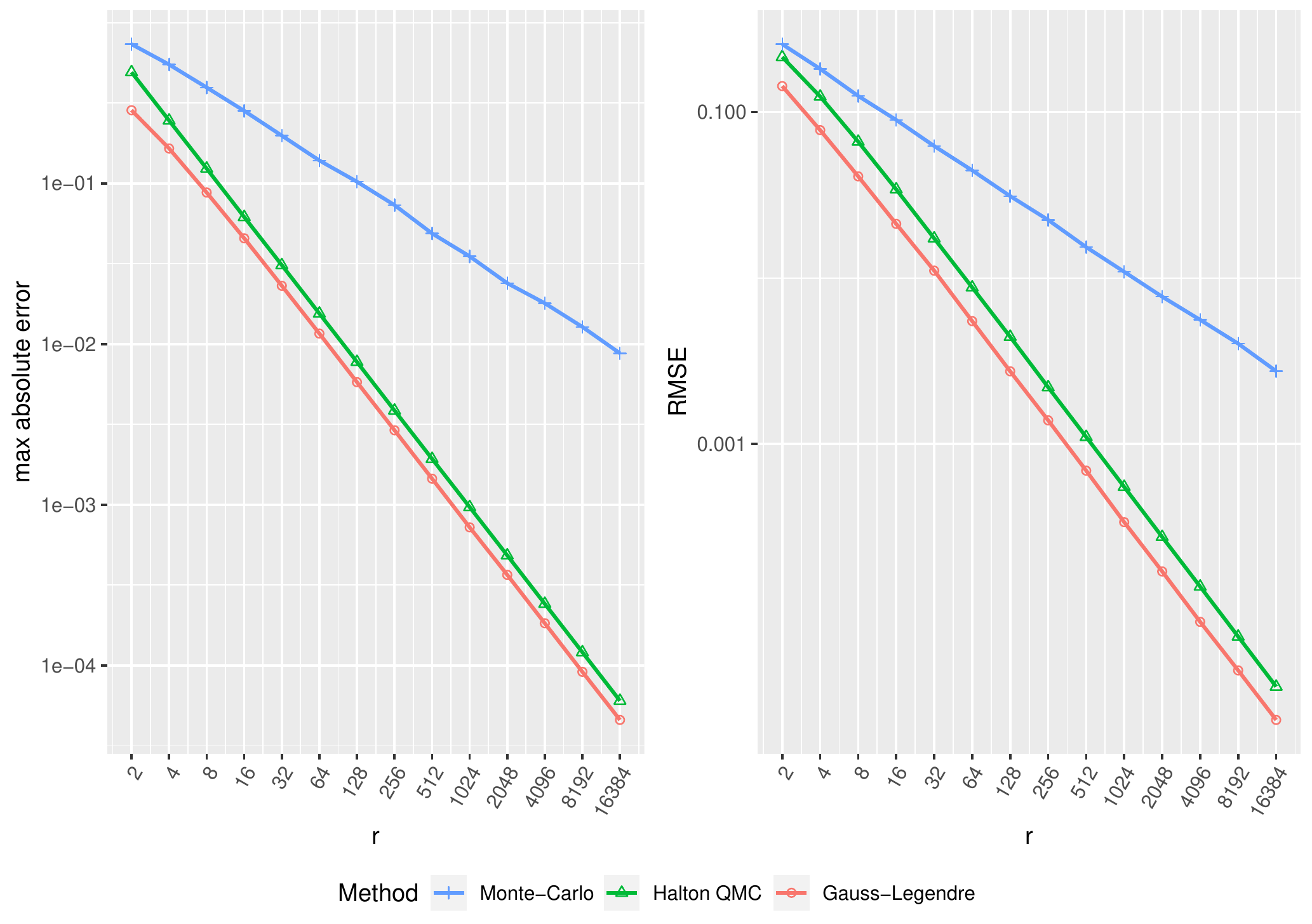}
\par\end{centering}
\caption{\label{Conv_fARS}Convergence behavior using different approximation methods for a non-smooth function $f\left(z\right)=\int_{\mathbb{R}} 1(x \leq z) \cdot g(x) \,\text{d}x$}
\end{figure}

The result for the non-smooth ARS confirms that the relative performance depends strongly on the smoothness of the approximated function. Here Monte-Carlo achieves the same convergence rate as before. The Halton draws and Gauss-Legendre achieve a higher rate, but Gauss-quadrature does not outperform the other two methods any more. 

As stated in section  \ref{sec_requiredAcc} the convergence rate of an approximation method translates into the required link-function $R(n)$ to achieve asymptotic normality. Figure \ref{Conv_RC} shows the results for the previous model $f(\cdot)$ from (\ref{model}) were the scaled approximation error $\mathcal{E}(r)$ (see equation (\ref{eq_error})) is plotted against the sample size in terms of Theorem 8 (ii). For Theorem 8 we require a link function which ensures $\plimn \sqrt{n}\mathcal{E}(R(n)) = 0$. For the constant link-function (upper left panel) this condition is not met by any approximation method. For the logarithmic link-function (upper right panel) and for the $\sqrt{n}$ link-function (lower left panel) only the more efficient Gauss-Hermite or the Gauss-Hermite and Halton draws, respectively, meet the condition. And for the computationally most costly link function $n$ (lower right panel) all methods meet the condition, even the Monte-Carlo method (blue line), albeit with a very flat slope.

\begin{figure}
\begin{centering}
\includegraphics[width=0.9\textwidth]{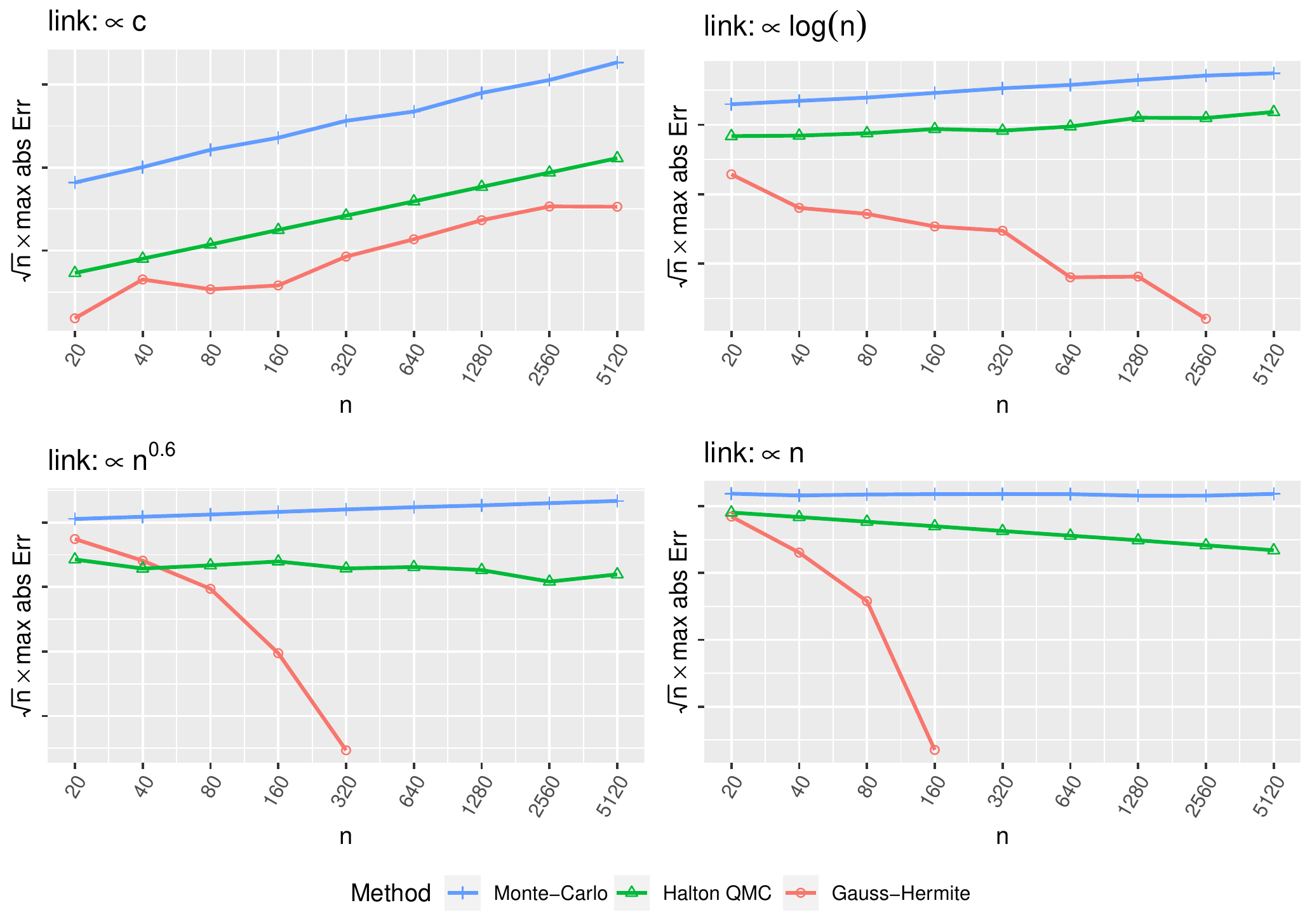}
\par\end{centering}
\caption{Convergence of $\sqrt{n}\mathcal{E}(R(n))$ for different link functions\label{Conv_RC}}
\end{figure}

\subsection*{Consequences for the practical application of MALE}
The overall estimation error ($\hat{\theta}_{MAL} - \theta_0$) in MAL has two components (compare also equation (\ref{ErrorComposition})). First, the sampling error and, second, the approximation error. In most practical  applications, the sample size is typically fixed and there is only a small variation allowed in the approximation accuracy. Thus, we now run a simulation were we generate $m=2000$ data sets for our random coefficient model with sample size fixed at $n=50$ and $n=5000$, respectively. We estimate the parameter $\bar\beta$ using the likelihood function (\ref{model}) from above. To illustrate the convergence behavior we compute the $RMSE = \sqrt{\frac{1}{m} \sum_{i=1}^{m} (\hat\theta_{MALE}-\theta_0)^2}$ based on the estimated parameter $\theta = \bar{\beta}$. Figure \ref{Conv_RC_fixedN} shows the result.

\begin{figure}
\begin{centering}
\includegraphics[width=0.9\textwidth]{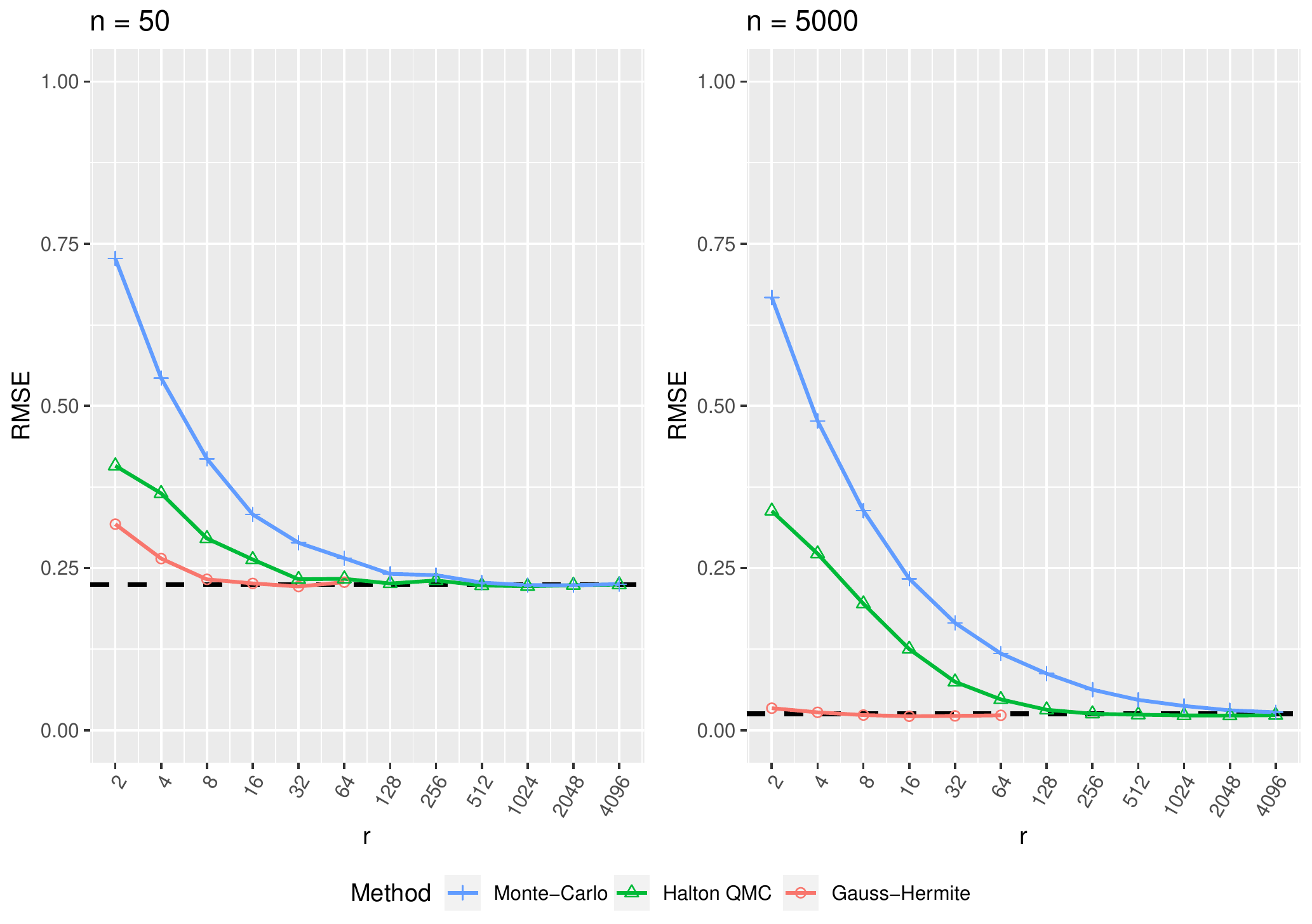}
\par\end{centering}
\caption{Convergence of the estimator with sample size $n$ fixed \label{Conv_RC_fixedN}}
\end{figure}
 
If there would be no approximation error only the sample error would cause variation in the result and would not be affected by $r$. This sample-error expressed in the $RMSE$ is given by the dashed horizontal line in Figure \ref{Conv_RC_fixedN}. For increasing sample size this error decreases (compare left and right panel). Indeed, the estimator's variance needed for standard hypothesis testing is just of this type.

For finite $r$ there is also an approximation error whose size depends on the approximation method and $r$. This error component increases the variance for the estimator and therefore potentially affects the interpretation of hypothesis tests. To correct for the approximation error in the inference one could apply corrections as proposed in \cite{kristensen2017}. So, also for practical applications, it is recommended to choose the most efficient method to minimize this error-component for a given computational cost. Moreover, Figure \ref{Conv_RC_fixedN} shows an increasing relevance of the approximation error due a rising number of observations. This leads to the conclusion that, for a large dataset, an appropriate approximation is of higher interest than for a small dataset.

\section{Summary and conclusions}

This paper discusses maximum approximated likelihood (MAL) estimators that generalize maximum simulated likelihood (MSL) estimators. A major advantage of MSL is that the underlying Monte Carlo simulation techniques provide favorable asymptotic properties under very general conditions. This not only makes it a versatile tool for the practitioner, it also simplifies the theoretical analysis. However, it has been frequently found that the computational costs required to achieve a sufficient approximation quality can be burdensome or infeasible. 
So it has become common practice to use more accurate numerical approximation algorithms such as quasi--Monte Carlo simulation, Gaussian quadrature, and integration on sparse grids.
This paper contributes to the theoretical underpinning of these approaches. It establishes sets of conditions on the model, the approximation approach and their interactions that ensure the consistency and the asymptotic efficiency of general MAL estimators. We also provide discussions of specific algorithms and models and show their behavior both asymptotically and in a simulation analysis. 

Overall, numerical approximation methods have stronger requirements on properties of the models such as the smoothness of the likelihood contributions than MSL. Given that these conditions are met, not only their finite sample properties, but also their asymptotic behavior can be superior. This is manifested mainly in the fact that the speed with which the computational burden has to be asymptotically increased  (as $n\rightarrow\infty$) can be dramatically decreased relative to MSL.

The verification of the general conditions provided in this paper  for specific models and approximation methods can be a nontrivial task. We provide examples to demonstrate the approach.
But we expect that more work discussing the specifics for different classes of models and algorithms would be useful for guiding the practitioner.

\bibliographystyle{cje}
\bibliography{library}

\appendix
\section{Appendix: Technical results} \label{app_thm}

In Section \ref{sec_tildef}, we have worked with distances of the log likelihood function from its quadrature-approximation as well as their derivatives. In order to deal with the logarithm in the log likelihood contribution, we now present a general result that relates the distance of a function or its derivatives up to order \(2\) to the respective distance between the logarithmized functions and its derivatives.

To this end, we will again use the notation \(D_{\bsalpha}^{(\bstheta)} g(\bstheta)\) to denote the partial derivatives with respect to \(\bstheta\), cf. \eqref{eqn_partial_deriv}. Moreover, we use \(\| \nabla_\bstheta g\|\) to denote the \(\ell_2\)-(vector-)norm of the gradient of \(g\) and \(\| \nabla_{\bstheta,\bstheta} g\|\) to denote the \(\ell_2\)-(matrix-)norm of the Hessian of \(g\).

Moreover, we will use the $L^\infty(X)$ space which contains all functions that are bounded on some domain $X$, i.e.
\[
	\|f\|_{L^\infty(X)} := \sup_{x \in X} |f(x)| < \infty .
\]

\begin{thm} \label{thm_concat}
 Consider real-valued functions \(g\) and \(h\) defined on a compact set \(\Theta \subset \R^p\). Assume that (i) the image of \(g\) and \(h\) is bounded away from zero, i.e. for all \(\bstheta \in \Theta\) it holds \(g(\bstheta), h(\bstheta) \in J = [\delta, D]\) with \(0 < \delta < D < \infty\); (ii) the partial derivatives of \(g\) and \(h\) up to order \(k \in \N_0\) exist and are bounded in \(L^\infty(\Theta)\).
Moreover, consider a function \(\psi: J \to \R\), whose first \(k+1\) derivatives exist and are bounded in \(L^\infty(J)\).
 Then, it holds for

 (i) \(k=0\), i.e. \(g,h \in L^\infty(\Theta)\) and \(\psi, \psi' \in L^\infty(J)\), that
 \begin{equation} \label{eqn_k0}
  \sup_{\bstheta \in \Theta} |\psi \circ g(\bstheta) - \psi \circ h(\bstheta) | \leq \|\psi'\|_{L^\infty(J)} \cdot \sup_{\bstheta \in \Theta} |g(\bstheta) - h(\bstheta)| .
 \end{equation}

 (ii) \(k=1\), i.e. \(g,h,\| \nabla_\bstheta g\|, \| \nabla_\bstheta h\| \in L^\infty(\Theta)\) and \(\psi, \psi', \psi'' \in L^\infty(J)\), that
 \begin{equation}\label{eqn_k1}
  \begin{aligned}
  \sup_{\bstheta \in \Theta} \| \nabla_\bstheta \psi \circ g (\bstheta) - \nabla_\bstheta \psi \circ h (\bstheta) \| & \leq \|\psi''\|_{L^\infty(J)} \cdot \sup_{\bstheta \in \Theta} |g(\bstheta) - h(\bstheta)| \cdot \sup_{\bstheta \in \Theta} \|\nabla_\bstheta h(\bstheta)\| \\
	& \quad \quad + \|\psi'\|_{L^\infty(J)}  \cdot \sup_{\bstheta \in \Theta} \|\nabla_\bstheta g(\bstheta) - \nabla_\bstheta h(\bstheta)\| .
  \end{aligned}
 \end{equation}

 (iii) \(k=2\), i.e. \(g,h,\| \nabla_\bstheta g\|, \| \nabla_\bstheta h\|, \| \nabla_{\bstheta,\bstheta} g\|, \| \nabla_{\bstheta,\bstheta} h\| \in L^\infty(\Theta)\) and \(\psi, \psi', \psi'', \psi''' \in L^\infty(J)\), that
 \begin{equation}\label{eqn_k2}
  \begin{aligned}
    \sup_{\bstheta \in \Theta} \| \nabla_{\bstheta \bstheta} \psi \circ g - \nabla_{\bstheta \bstheta} \psi \circ h \| & \leq \|\psi'''\|_{L_\infty(J)} \sup_{\bstheta \in \Theta} \| \nabla_\bstheta h(\bstheta) \|^2  \sup_{\bstheta \in \Theta} | g(\bstheta) - h(\bstheta) | \\
   & \quad \quad + \|\psi'\|_{L_\infty(J)} \sup_{\bstheta \in \Theta} \big\| \nabla_{\bstheta \bstheta} g(\bstheta) - \nabla_{\bstheta \bstheta} h(\bstheta) \big\|  \\
   & \quad \quad + \|\psi''\|_{L_\infty(J)} \sup_{\bstheta \in \Theta} \| \nabla_{\bstheta \bstheta} h(\bstheta) \|  \sup_{\bstheta \in \Theta} | g(\bstheta) - h(\bstheta) | \\
   & \quad\quad + 2 \|\psi''\|_{L_\infty(J)} \sup_{\bstheta \in \Theta}\| \nabla_\bstheta h(\bstheta) \|   \sup_{\bstheta \in \Theta} \| \nabla_\bstheta g(\bstheta) - \nabla_\bstheta h(\bstheta) \| \\
   & \quad \quad  + \|\psi''\|_{L_\infty(J)} \sup_{\bstheta \in \Theta} \| \nabla_\bstheta g(\bstheta) - \nabla_\bstheta h(\bstheta) \|^2 .
  \end{aligned}
 \end{equation}
\end{thm}

The theorem is proven by the following three Lemmas.

\begin{lemma} \label{lemma_mvt}
For a differentiable function \(\psi: J \to \R\), with \([a,b] \subseteq J \subset \R\), it holds that
\[
	|\psi(b) - \psi(a)| \leq \|\psi'\|_{L_\infty} |b-a| .
\]	

\begin{proof}
	The claim follows immediately from the mean-value theorem, i.e. there exists \(\xi \in [a,b]\) such that it holds
	\[\psi'(\xi) = \frac{\psi(b)-\psi(a)}{b-a} .\]
\end{proof}
\end{lemma}

Lemma \ref{lemma_mvt} immediately implies
\[
 |\psi \circ g(\bstheta) - \psi \circ h(\bstheta) | \leq \|\psi'\|_{L_\infty(J)} | g(\bstheta) - h(\bstheta) |,
\]
which implies Theorem \ref{thm_concat} for the case \(k=0\), i.e. \eqref{eqn_k0}.

Next, we deal with the gradient.

\begin{lemma}
 For a twice differentiable function \(\psi: J \to \R\) with \(\psi, \psi', \psi'' \in L^\infty(J)\) and $J$-valued functions \(g,h: \Theta \to J\) with bounded derivatives of first order it holds that
 \begin{equation*} \label{eqn_concat_grad}
 \begin{aligned}
 	\sup_{\bstheta \in \Theta} \| \nabla_{\bstheta} \psi \circ g(\bstheta) - \nabla_{\bstheta} \psi \circ h(\bstheta) \| \leq &  \|\psi''\|_{L_\infty(J)} \sup_{\bstheta \in \Theta} | g(\bstheta) - h(\bstheta) | \|\nabla_\bstheta h(\bstheta)\| \\
 	& \quad + \|\psi'\|_{L_\infty(J)} \sup_{\bstheta \in \Theta} \| \nabla_{\bstheta} g(\bstheta) - \nabla_{\bstheta} h(\bstheta) \| .
 \end{aligned}
 \end{equation*}

 \begin{proof}
 Using the chain rule for the gradient \(\nabla_\bstheta \psi \circ g(\bstheta) = \psi'(g(\bstheta)) \cdot \nabla g(\bstheta)\) we compute
 \begin{align*}
  \big\| \nabla_\bstheta \psi \circ g(\bstheta) - \nabla_\bstheta \psi \circ h(\bstheta) \big\| & = \big\| \psi' \circ g(\bstheta) \nabla_\bstheta g(\bstheta) - \psi' \circ h(\bstheta) \nabla_\bstheta h(\bstheta) \big\| \\
	& = \big\| \psi' \circ g(\bstheta) \nabla_\bstheta g(\bstheta) - \psi' \circ g(\bstheta) \nabla_\bstheta h(\bstheta) \\
	& \quad \quad + \psi' \circ g(\bstheta) \nabla_\bstheta h(\bstheta) - \psi' \circ h(\bstheta) \nabla_\bstheta h(\bstheta) \big\| \\
	& \leq | \psi' \circ g(\bstheta) | \cdot \big\| \nabla_\bstheta g(\bstheta) - \nabla_\bstheta h(\bstheta) \big\| \\
	& \quad\quad + | \psi' \circ g(\bstheta) - \psi' \circ h(\bstheta)| \cdot \big\| \nabla_\bstheta h(\bstheta) \big\| .
 \end{align*}

 By \(| \psi' \circ g(\bstheta) | \leq \|\psi'\|_{L_\infty(J)}\) and by applying Lemma \ref{lemma_mvt} we obtain
 \[
  | \psi' \circ g(\bstheta) - \psi' \circ h(\bstheta)| \leq \|\psi''\| _{L_\infty(J)} |g(\bstheta) - h(\bstheta)| .
 \]
 Taking the \(\sup\) over all \(\bstheta \in \Theta\) we arrive at the desired result, which is \eqref{eqn_k1}, i.e. (ii) in Theorem \ref{thm_concat}.

 \end{proof}
\end{lemma}

Next, we deal with the Hessian matrix \(\nabla_{\bstheta \bstheta} f(\bstheta) = (\partial_{\theta_k} \partial_{\theta_j} f(\bstheta))_{k,j=1}^p\) .

\begin{lemma}
For a \(3\)-times differentiable function \(\psi: J \to \R\) with \(\psi, \psi', \psi'', \psi''' \in L^\infty(J)\) and \(J\)-valued functions \(g,h: \Theta \to J\) with bounded derivatives up to second order it holds that
 \begin{equation} \label{eqn_concat_hessian}
 \begin{aligned}
  \sup_{\bstheta \in \Theta} \| \nabla_{\bstheta \bstheta} \psi \circ g - \nabla_{\bstheta \bstheta} \psi \circ h \| & \leq \|\psi'''\|_{L_\infty(J)} \sup_{\bstheta \in \Theta} \| \nabla_\bstheta h(\bstheta) \|^2  | \sup_{\bstheta \in \Theta} g(\bstheta) - h(\bstheta) | \\
   & \quad \quad + \|\psi'\|_{L_\infty(J)} \big\| \sup_{\bstheta \in \Theta} \nabla_{\bstheta \bstheta} g(\bstheta) - \nabla_{\bstheta \bstheta} h(\bstheta) \big\|  \\
   & \quad \quad + \|\psi''\|_{L_\infty(J)} \sup_{\bstheta \in \Theta} \| \nabla_{\bstheta \bstheta} h(\bstheta) \| | g(\bstheta) - h(\bstheta) | \\
   & \quad\quad + 2 \|\psi''\|_{L_\infty(J)} \sup_{\bstheta \in \Theta}\| \nabla_\bstheta h(\bstheta) \|   \sup_{\bstheta \in \Theta} \| \nabla_\bstheta g(\bstheta) - \nabla_\bstheta h(\bstheta) \| \\
   & \quad \quad  + \|\psi''\|_{L_\infty(J)} \sup_{\bstheta \in \Theta} \| \nabla_\bstheta g(\bstheta) - \nabla_\bstheta h(\bstheta) \|^2 .
 \end{aligned}
 \end{equation}

 \begin{proof}
 First we note that it holds by the multivariate chain-rule that
 \begin{equation}
 \begin{aligned}
 \nabla_{\bstheta \bstheta} \psi \circ g(\bstheta) & = \psi'' \circ g(\bstheta) \, (\nabla_\bstheta g(\bstheta) ) \cdot (\nabla_\bstheta g(\bstheta))^t + \psi' \circ g(\bstheta) \nabla_{\bstheta \bstheta} g(\bstheta) \\
  & = \psi'' \circ g(\bstheta) \cdot p_\bstheta(g) + \psi' \circ g(\bstheta) \cdot \nabla_{\bstheta \bstheta} g(\bstheta) ,
 \end{aligned}
\end{equation}
where we abbreviated \(p_\bstheta(g) := (\nabla_\bstheta g(\bstheta) ) \cdot (\nabla_\bstheta g(\bstheta))^t\).
Here, \(p_\bstheta(g) \in \R^{p \times p}\) is a rank-1 matrix that consists of the outer product of the gradient of \(g\) at \(\bstheta\) with itself. For its norm it holds that \(\|p_\bstheta(g)\| = \|\nabla_\bstheta g(\bstheta)\|^2\). Moreover, \(\nabla_{\bstheta \bstheta} g(\bstheta)\) denotes the Hessian matrix of \(g\) at \(\bstheta\).

Now, we proceed analogously to the proof of the preceeding Lemma and compute

\begin{equation} \label{eqn_comppp}
 \begin{aligned}
  & \quad \| \nabla_{\bstheta \bstheta} \psi \circ g(\bstheta) - \nabla_{\bstheta \bstheta} \psi \circ h(\bstheta) \| \\
  & = \big\| \psi'' \circ g(\bstheta) \, p_\bstheta(g) + \psi' \circ g(\bstheta) \nabla_{\bstheta \bstheta} g(\bstheta) - \psi'' \circ h(\bstheta) \, p_\bstheta(h) - \psi' \circ h(\bstheta) \nabla_{\bstheta \bstheta} h(\bstheta) \big\| \\
  & = \big\| \psi'' \circ g(\bstheta) \, p_\bstheta(g) - \psi'' \circ h(\bstheta) \, p_\bstheta(h) + \psi'' \circ g(\bstheta) \, p_\bstheta(h) - \psi'' \circ g(\bstheta) \, p_\bstheta(h) \\
  & \quad \quad + \psi' \circ g(\bstheta) \nabla_{\bstheta \bstheta} g(\bstheta) - \psi' \circ h(\bstheta) \nabla_{\bstheta \bstheta} h(\bstheta) \big\| \\
  & \leq |\psi'' \circ g(\bstheta)| \big\| p_\bstheta(g) - p_\bstheta(h) \| + \| p_\bstheta(h) \| | \psi'' \circ g(\bstheta) - \psi'' \circ h(\bstheta) |   \\
  & \quad \quad + \big \| \psi' \circ g(\bstheta) \nabla_{\bstheta \bstheta} g(\bstheta) - \psi' \circ h(\bstheta) \nabla_{\bstheta \bstheta} h(\bstheta) \big\| \\
  & \leq \|\psi''\|_{L_\infty(J)} \big\| p_\bstheta(g) - p_\bstheta(h) \big\| + \| \nabla_\bstheta h(\bstheta) \|^2 \| \psi'''\|_{L_\infty(J)} | g(\bstheta) - h(\bstheta) |   \\
  & \quad \quad + \big \| \psi' \circ g(\bstheta) \nabla_{\bstheta \bstheta} g(\bstheta) - \psi' \circ h(\bstheta) \nabla_{\bstheta \bstheta} h(\bstheta) \big\| .
 \end{aligned}
\end{equation}

 Next, we will take care of the last summand in \eqref{eqn_comppp} and derive
 \begin{equation*}
  \begin{aligned}
  & \| \psi' \circ g(\bstheta) \nabla_{\bstheta \bstheta} g(\bstheta) - \psi' \circ h(\bstheta) \nabla_{\bstheta \bstheta} h(\bstheta) \big\| \\
  & = \big \| \psi' \circ g(\bstheta) \nabla_{\bstheta \bstheta} g(\bstheta) - \psi' \circ g(\bstheta) \nabla_{\bstheta \bstheta} h(\bstheta) + \psi' \circ g(\bstheta) \nabla_{\bstheta \bstheta} h(\bstheta) - \psi' \circ h(\bstheta) \nabla_{\bstheta \bstheta} h(\bstheta) \big\| \\
  & \leq |\psi' \circ g(\bstheta)| \big\| \nabla_{\bstheta \bstheta} g(\bstheta) - \nabla_{\bstheta \bstheta} h(\bstheta) \big\| + \| \nabla_{\bstheta \bstheta} h(\bstheta) \| \Big| \psi' \circ g(\bstheta) - \psi' \circ h(\bstheta) \Big| \\
  & \leq \|\psi'\|_{L_\infty(J)} \big\| \nabla_{\bstheta \bstheta} g(\bstheta) - \nabla_{\bstheta \bstheta} h(\bstheta) \big\| + \| \nabla_{\bstheta \bstheta} h(\bstheta) \| \|\psi''\|_{L_\infty(J)} | g(\bstheta) - h(\bstheta) | .
  \end{aligned}
 \end{equation*}

 Inserting this back into \eqref{eqn_comppp} we obtain
 \begin{equation} \label{eqn_comppp2} 
  \begin{aligned}
   & \quad \| \nabla_{\bstheta \bstheta} \psi \circ g(\bstheta) - \nabla_{\bstheta \bstheta} \psi \circ h(\bstheta) \| \\
   & \leq \|\psi''\|_{L_\infty(J)} \big\| p_\bstheta(g) - p_\bstheta(h) \big\| + \| \nabla_\bstheta h(\bstheta) \|^2 \| \psi'''\|_{L_\infty(J)} | g(\bstheta) - h(\bstheta) |   \\
   & \quad \quad + \|\psi'\|_{L_\infty(J)} \big\| \nabla_{\bstheta \bstheta} g(\bstheta) - \nabla_{\bstheta \bstheta} h(\bstheta) \big\| + \| \nabla_{\bstheta \bstheta} h(\bstheta) \| \|\psi''\|_{L_\infty(J)} | g(\bstheta) - h(\bstheta) | .
  \end{aligned}
 \end{equation}

It remains to bound the term \(\big\| p_\bstheta(g) - p_\bstheta(h) \big\|\) in \eqref{eqn_comppp}. To this end, we derive for vectors \(\bsv, \bsw \in \R^p\)
\begin{equation} \label{eqn_vector_diff}
 \begin{aligned}
 \| \bsv \cdot \bsv^t - \bsw \cdot \bsw^t \| &= \| \bsv \cdot \bsv^t - \bsv \cdot \bsw^t + \bsv \cdot \bsw^t - \bsw \cdot \bsw^t \| \\
	   & = \| \bsv \cdot (\bsv^t - \bsw^t) + (\bsv - \bsw) \cdot \bsw^t \| \\
	   & \leq \|\bsv - \bsw\| \ \left(\|\bsv\| + \|\bsw\| \right) \\
	   & = \|\bsv - \bsw\| \ \left(\|\bsw + (\bsv - \bsw)\| + \|\bsw\| \right) \\
	   & \leq \|\bsv - \bsw\| \ \left(\|\bsw\| + \|\bsv - \bsw\| + \|\bsw\| \right) \\
	   & = 2 \|\bsw\| \|\bsv - \bsw\| + \|\bsv - \bsw\|^2 .
 \end{aligned}
\end{equation}
Now, using \eqref{eqn_vector_diff} with \(\bsv = \nabla_\bstheta g(\bstheta) \) and \(\bsw = \nabla_\bstheta h(\bstheta)\) we obtain
\begin{equation}
 \begin{aligned} \label{eqn_grad_diff}
 \big\| p_\bstheta(g) - p_\bstheta(h) \big\| & = \big\| (\nabla_\bstheta g(\bstheta) ) \cdot (\nabla_\bstheta g(\bstheta))^t - (\nabla_\bstheta h(\bstheta) ) \cdot (\nabla_\bstheta h(\bstheta))^t \big\| \\
	  & \leq 2 \| \nabla_\bstheta h(\bstheta) \| \nabla_\bstheta g(\bstheta) - \nabla_\bstheta h(\bstheta) \| + \| \nabla_\bstheta g(\bstheta) - \nabla_\bstheta h(\bstheta) \|^2 .
 \end{aligned}
\end{equation}

Hence, inserting \eqref{eqn_grad_diff} into \eqref{eqn_comppp2} we arrive at
\begin{equation*}
 \begin{aligned}
 & \quad \| \nabla_{\bstheta \bstheta} \psi \circ g(\bstheta) - \nabla_{\bstheta \bstheta} \psi \circ h(\bstheta) \| \\
 & \leq \| \nabla_\bstheta h(\bstheta) \|^2 \|\psi'''\|_{L_\infty(J)} | g(\bstheta) - h(\bstheta) | + \|\psi'\|_{L_\infty(J)} \big\| \nabla_{\bstheta \bstheta} g(\bstheta) - \nabla_{\bstheta \bstheta} h(\bstheta) \big\|  \\
   & \quad \quad + \| \nabla_{\bstheta \bstheta} h(\bstheta) \| \|\psi''\|_{L_\infty(J)} | g(\bstheta) - h(\bstheta) | + \|\psi''\|_{L_\infty(J)} 2 \| \nabla_\bstheta h(\bstheta) \| \nabla_\bstheta g(\bstheta) - \nabla_\bstheta h(\bstheta) \| \\
   & \quad \quad  + \|\psi''\|_{L_\infty(J)} \| \nabla_\bstheta g(\bstheta) - \nabla_\bstheta h(\bstheta) \|^2 ,
 \end{aligned}
\end{equation*}

 which concludes the proof.

 \end{proof}

\end{lemma}

For the special case \(\psi(x) = \log(x)\) we obtain the following result.

\begin{corollary} \label{cor_concat}
 Consider real-valued functions \(g,h\) on some compact domain \(\Theta \subset \R^p\). Assume that the image of \(g\) and \(h\) is bounded away from zero, i.e. for all \(\bstheta \in \Theta\) it holds \(g(\bstheta), h(\bstheta) \in J = [\delta, D]\) with \(0 < \delta < 1 \leq D < \infty\). Then

 (i) for \(g,h \in L^\infty(\Theta)\) it holds
 \begin{equation}
  \sup_{\bstheta \in \Theta} |\log g(\bstheta) - \log h(\bstheta)| \leq \frac{1}{\delta} \sup_{\bstheta \in \Theta} |g(\bstheta) - h(\bstheta)| .
 \end{equation}

 (ii) for differentiable \(g,h\) with bounded first derivatives it holds
  \begin{equation}
   \begin{aligned}
    & \sup_{\bstheta \in \Theta} \| \nabla_{\bstheta} \log g(\bstheta) - \nabla_{\bstheta} \log h(\bstheta)\| \\
    & \quad \quad \leq C_1(h) \left( \sup_{\bstheta \in \Theta} |g(\bstheta) - h(\bstheta)| + \sup_{\bstheta \in \Theta} \| \nabla_{\bstheta} g(\bstheta) - \nabla_{\bstheta} h(\bstheta)\| \right),
   \end{aligned}
  \end{equation}
where
\[
C_1(h) = \frac{1 + \sup_{\bstheta} \|\nabla_\bstheta h\|}{\delta^2} .
\]
 
 (iii) for \(2\)-times differentiable \(g,h\) with bounded derivatives up to second order it holds
  \begin{equation}
   \begin{aligned}
    & \sup_{\bstheta \in \Theta} \| \nabla_{\bstheta,\bstheta} \log g(\bstheta) - \nabla_{\bstheta,\bstheta} \log h(\bstheta)\| \\
    & \leq  C_2(h) \Bigg( \sup_{\bstheta \in \Theta} |g(\bstheta) - h(\bstheta)| + \sup_{\bstheta \in \Theta} \| \nabla_{\bstheta} g(\bstheta) - \nabla_{\bstheta} h(\bstheta)\| \\
    & \quad \quad + \sup_{\bstheta \in \Theta} \| \nabla_{\bstheta} g(\bstheta) - \nabla_{\bstheta} h(\bstheta)\|^2 + \sup_{\bstheta \in \Theta} \| \nabla_{\bstheta,\bstheta} g(\bstheta) - \nabla_{\bstheta,\bstheta} h(\bstheta)\| \Bigg) ,
   \end{aligned}
  \end{equation}
where
\[
C_2(h) = 4 \frac{1 + \sup_{\bstheta} \|\nabla_\bstheta h\|^2 + \sup_{\bstheta} \|\nabla_{\bstheta,\bstheta} h \| }{\delta^3} .
\]
 
\begin{proof}
First note that  
\[\frac{\rd^j}{\rd x^j}\log(x) = \frac{(-1)^{j-1} (j-1)!}{ x^{j} }. \]
Then, for (i) and (ii) we use that \(\|\psi^{(j)}\|_{L_\infty(J)} = \frac{(j-1)!}{\delta^j}\) with \(j \in \N\) and \(J = [\delta, \infty)\). For (iii) we additionally use the inequality \(\max(x,x^2) \leq 1 + x^2\).
\end{proof}

\end{corollary}

We finally prove the following result that was employed in the analysis of the logit model in Section \ref{s_examples}.

\begin{lemma} \label{cor_logit_derivative}
 Let \(\bstheta, \bsv \in \R^n\) and
 \begin{equation}
 f(\bstheta, \bsv) = \frac{1}{ 1 + \exp(- \bstheta \cdot \bsv) } .
 \end{equation}
 Then it holds for all \(\bsalpha, \bsbeta \in \N_0^n\) that 
 \begin{enumerate}
 	\item[(i)] there exists a constant \(0 < c(\bsalpha, \bsbeta) < \infty\) such that
 \begin{equation}
 | D^{(\bstheta)}_\bsalpha D^{(\bsv)}_\bsbeta f(\bstheta, \bsv) | \leq c(\bsalpha, \bsbeta) \bstheta^\bsbeta \cdot \bsv^\bsalpha \quad \text{ for all } \bsv, \bstheta \in \R^n .
 \end{equation}
 
 \item[(ii)] there exists a constant \(0 < \tilde{c}(\bsalpha, \bsbeta) < \infty\) such that
 \begin{equation}
 | D^{(\bstheta)}_\bsalpha D^{(\bsv)}_\bsbeta f(\bstheta, \bsv) | \leq \tilde{c}(\bsalpha, \bsbeta) \bstheta^\bsbeta \cdot \frac{e^{\frac{\bsv^t \bsv}{2}}}{\prod_{i=1}^n \sqrt{1+v_i^2}} \quad \text{ for all } \bsv, \bstheta \in \R^n .
 \end{equation}
 \end{enumerate}

 \begin{proof}
First, we note that (ii) follows from (i), because the function $t \mapsto e^{t^2/2}/\sqrt{1+t^2}$ grows faster than any polynomial.

In order to prove (i), we write $f(\bstheta, \bsv) := g(\bstheta \cdot \bsv)$, where $g(t) = (1+e^{-t})^{-1}$. Note that $g$ and hence $g^p$ are bounded on $\R$ for all $p \in \N$.
Next, we use that
\begin{equation}
	\frac{\partial}{\partial\theta_i} g(\bstheta \cdot \bsv) = v_i ( g(\bstheta \cdot \bsv) -  g^2(\bstheta \cdot \bsv))
\end{equation}
and likewise also $\frac{\partial}{\partial v_i} g(\bstheta \cdot \bsv) = \theta_i ( g(\bstheta \cdot \bsv) -  g^2(\bstheta \cdot \bsv))$.
This proves (i) for $|\bsalpha|=1$ and $|\bsbeta|=1$. The case of higher order derivatives follows by induction by showing that $D^{(\bstheta)}_\bsalpha D^{(\bsv)}_\bsbeta g(\bstheta \cdot \bsv)$ always has a representation as $\sum_j c_j \bsv^{\bsa_j} \bstheta^{\bsb_j} g^{p_j}(\bstheta \cdot \bsv)$ with $\bsa_j \leq \bsalpha$, $\bsb_j \leq \bsalpha$ and $p_j \leq |\bsalpha| + |\bsbeta| + 1$.
 \end{proof}
\end{lemma}

\end{document}